\documentclass[a4paper,leqno]{amsart}

\usepackage{latexsym}
\usepackage[english]{babel}
\usepackage{fancyhdr}
\usepackage[mathscr]{eucal}
\usepackage{amsmath}
\usepackage{mathrsfs}
\usepackage{amsthm}
\usepackage{amsfonts}
\usepackage{amssymb}
\usepackage{amscd}
\usepackage{bbm}
\usepackage{graphicx}
\usepackage{subcaption}
\usepackage{graphics}
\usepackage{latexsym}
\usepackage{color}
\usepackage{float}

\usepackage{pifont}
\usepackage{booktabs} % for '\midrule' macro

\newcommand{\ud}{\mathrm{d}}

\newcommand{\ii}{\mathrm{i}}
\newcommand{\cH}{\mathcal{H}}

\theoremstyle{plain}
\newtheorem{theorem}{Theorem}[section]
\newtheorem{lemma}[theorem]{Lemma}

\newtheorem{proposition}[theorem]{Proposition}

\theoremstyle{definition}

\newtheorem{remark}[theorem]{Remark}
\newtheorem{example}[theorem]{Example}

\numberwithin{equation}{section}

\begin{document}

\title[Self-adjointness in Quantum Mechanics: a pedagogical path]
{Self-adjointness in Quantum Mechanics: \\ a pedagogical path}
\author[A.~Cintio]{Andrea Cintio}
\address[A.~Cintio]{Institute for Chemical and Physical processes, National Research Council (CNR) \\ via G.~Moruzzi 1 \\ I-56124 Pisa (ITALY).}
\email{andrea.cintio@pi.ipcf.cnr.it}
\author[A.~Michelangeli]{Alessandro Michelangeli}
\address[A.~Michelangeli]{Institute for Applied Mathematics and Hausdorff Center for Mathematics, University of Bonn \\ Endenicher Allee 60 \\ 
D-53115 Bonn (GERMANY).}
\email{michelangeli@iam.uni-bonn.de}

%\dedicatory{}

\begin{abstract}
 Observables in quantum mechanics are represented by self-adjoint operators on Hilbert space. Such ubiquitous, well-known, and very foundational fact, however, is traditionally subtle to be explained in typical first classes in quantum mechanics, as well as to senior physicists who have grown up with the lesson that self-adjointness is ``just technical''. The usual difficulties are to clarify the connection between the demand for certain physical features in the theory and the corresponding mathematical requirement of self-adjointness, and to distinguish between self-adjoint and hermitian operator not just at the level of the mathematical definition but most importantly from the perspective that mere hermiticity, without self-adjointness, does not ensure the desired physical requirements and leaves the theory inconsistent. In this work we organise an amount of standard facts on the physical role of self-adjointness into a coherent pedagogical path aimed at making quantum observables emerge as necessarily self-adjoint, and not merely hermitian operators. Next to the central core of our line of reasoning -- the necessity of a non-trivial declaration of a domain to associate with the formal action of an observable, and the emergence of self-adjointness as a consequence of fundamental physical requirements -- we include some complementary materials consisting of a few instructive mathematical proofs and a short retrospective, ranging from the past decades to the current research agenda, on the self-adjointness problem for quantum Hamiltonians of relevance in applications. 
\end{abstract}

\date{\today}

\subjclass[2000]{}

\thanks{This work is partially supported by the Alexander von Humboldt Foundation. The authors warmly acknowledge V.~Bashmakov, M.~Gallone, C.~Kodarin, and R.~Scandone for the many fruitful discussions on the subject}
% grant ``\emph{Cond-Math: Condensed Matter and Mathematical Physics}'' code RBFR13WAET}

\maketitle

%\tableofcontents

\section{Introduction}\label{intro}

Quantum mechanics is a central, mandatory topic in virtually all undergraduate programmes for physicists around the world.
%(in comparison, general relativity, the other major leg of modern physics, is mainly taught at an optional level). 

Physicists in the course of their training are exposed to various degrees of details concerning the mathematical structure of quantum mechanics: this always includes the notion of hermiticity (or symmetry -- we shall consider them as synonymous, as customary) of the operators associated with physical observables, as well as the proof that expectations of hermitian operators are indeed real numbers. The (simple) proof of the inverse implication, namely that by polarisation an operator with real expectations is necessarily hermitian, is often omitted, yet the association of quantum observables with hermitian operators remains physically well grounded and part of the background of any physicist. More seldom it is mentioned that quantum observables are actually self-adjoint, and not merely hermitian operators, yet the two concepts of hermiticity and self-adjointness are usually kept on an equal footing and often used haphazardly, as if they were the same. In other cases self-adjointness is more properly introduced in class, but under the perspective that such extra requirement boils down to mathematical technicalities, possibly of physical relevance, yet not worth being worked out explicitly, the main physical content being the reality of expectation of hermitian operators.

Quantum mechanics is of course learnt also by many mathematicians, usually along the reversed approach from the mathematical axioms to the applications. Self-adjointness, from this perspective, is introduced through its plain mathematical definition, the theory of self-adjoint operators on Hilbert space is developed, and consequences in application to quantum mechanics are derived. This route is clean, but may obfuscate the physical motivation to self-adjointness, and above all it is harder to access for physics undergrads exposed to a first class in quantum mechanics, as well as for those senior physicists who have grown up with the lesson that self-adjointness is ``just technical''.

In this work we propose a pedagogical path, ideally addressed to both the above categories of physics undergrads and professional physicists, which makes the notion of self-adjointness emerge in association with quantum observables in a way that be accessible, mathematically rigorous, physically deep-rooted, and eventually stringent, in the sense that it does not leave room to dismissing the subject as a mere technicality if one wants to develop meaningful physics. These features should also make our line of reasoning appealing for mathematicians approaching quantum mechanics: they certainly do not have difficulties in digesting the definition of self-adjoint operator, but could appreciate seeing how the notion gets shaped in connection to various fundamental physical requirements.

Of course, tacitly speaking we imagine our readership consisting of those experts of the subject that are supposed to teach such topics in class or to their graduate students.

Thus, to stress our perspective, \emph{we are not taking the abstract point of view of the mathematical foundations of quantum mechanics} \cite{Dirac-PrinciplesQM,vonNeumann-MathFoundQM,Mackey-QM-1963-2004,Strocchi-MathQM,DellAnt_QM1-2015} \emph{or the general theory of self-adjoint operators in Hilbert space} \cite{Blank-Exner-Havlicek-2008,Amrein-HilberSpMethods-2009,schmu_unbdd_sa}. 
% 
% focusing here on an efficient route 
% 
% In this respect our natural starting point 
% Our natural starting point is rather operational in nature: we place ourselves at the precise stage at which self-adjointness kicks in in the course of the typical physical discussion of the mathematical framework of quantum mechanics, namely the emergence of quantum observables through `first quantisation' from classical mechanics, and we intend to focus on an efficient route to 
% 
Our natural starting point is rather operational in nature: we move from the precise stage at which quantum observables are introduced as linear and (at least) hermitian operators acting within the Hilbert space of states of the considered quantum system, as done in the most established and traditional physical introductions on the mathematical framework of quantum mechanics -- we have in mind, among others, Dirac \cite{Dirac-PrinciplesQM}, Landau \cite{Landau-Lifshitz-3}, Cohen-Tannoudji et al.~\cite{Cohen-Tannoudji-1977-2020}, Sakurai \cite{sakurai_napolitano_2017}, Weinberg \cite{weinberg_2015} -- and we intend to focus on a conceptually efficient route to make self-adjointness (and not just mere hermiticity) kick in for such operators. The playground we find most instructive from this perspective are those operators emerging through `first quantisation' from classical mechanics: considerations that are made in concrete for this class can be then used in more abstract settings.

First quantisation, rigorously speaking, is a nebulous concept (at a more fundamental level position and momentum operators emerge as generators in the Schr\"{o}dinger representation of the Weyl $C^*$-algebra \cite[Chapter 3]{Strocchi-MathQM}): at this point we just refer to it quite pragmatically, as is indeed done in a typical first introduction to quantum mechanics for physicists. It provides a physically grounded, operational recipe to construct quantum observables from their classical counterparts (up to non-commutative ordering) in the form of linear functional operators acting in the Hilbert space of states for the considered quantum system \cite[Sect.~I.2]{vonNeumann-MathFoundQM}, \cite[\S 15 and \S 17]{Landau-Lifshitz-3}, \cite[Sect.~1.3, 1.4, 3.3]{weinberg_2015}, \cite[Sect.~III.5]{Cohen-Tannoudji-1977-2020}, \cite[Sect.~1.6]{sakurai_napolitano_2017}.

Non-restrictively for our purposes, we shall often consider a quantum system with one spatial degree of freedom, hence a quantum particle of mass $m$ in one dimension, with Hilbert space $L^2(\mathbb{R})$ (or $L^2(a,b)$ for a particle in the box $(a,b)$, and similar choices): then first quantisation associates the quantum position observable with the multiplication by the spatial coordinate $x$ and the quantum momentum observable with the differential operator $-\ii\hbar\frac{\ud}{\ud x}$. (We keep the one-dimensional setting for many of our examples having in mind that this is the first playground used in class, but of course the whole material can be re-phrased in arbitrary dimension.) This allows one to pass from a classical Hamiltonian function $H(p,q)$ on $\mathbb{R}^2$ to the quantum Hamiltonian counterpart $H(-\ii\hbar\frac{\ud}{\ud x},x)$ on $L^2(\mathbb{R})$ (up to ordering, as said, due to the non-commutativity of position and momentum, although the above correspondence is unambiguous for classical Hamiltonians $H(p,q)=\frac{1}{2m}p^2+V(q)$). Thus, in practice, we shall discuss the emergence of the notion of self-adjointness having in mind usual quantum observables such as
\begin{equation}\label{eq:op-list}
 \begin{array}{ccl}
  \textrm{multiplication by $x$} & & \textrm{(position operator)} \\
  \displaystyle-\ii\frac{\ud}{\ud x}\qquad  & & \textrm{(momentum operator)} \\
  \quad \displaystyle-\frac{\ud^2}{\ud x^2} & & \textrm{(kinetic energy operator)} \\
  \displaystyle\Big(-\ii\frac{\ud}{\ud x}-A(x)\Big)^2 & & \textrm{(magnetic kinetic energy operator)} \\
  \displaystyle-\frac{\ud^2}{\ud x^2}+V(x)\qquad  & & \textrm{(Schr\"{o}dinger operator)} \\
  \sqrt{\displaystyle-\frac{\ud^2}{\ud x^2}+1}+V(x) & & \textrm{(semi-relativistic Schr\"{o}dinger operator)}
 \end{array}
\end{equation}
and so forth, where inessential physical constants have been re-scaled out.

At this stage one encounters also observables that are inherently quantum (i.e., with no classical analogue), and hence are not the outcome of first quantisation, such as the spin: but such observables actually correspond to $n\times n$ matrices acting on $\mathbb{C}^n$ for some $n\in\mathbb{N}$, for which the request of hermiticity already ensures (is in fact equivalent to) self-adjointness.

For a vast part of the physical discussion on the general principles of quantum mechanics, observables like \eqref{eq:op-list} are introduced as formal operators acting on square-integrable functions (`wave functions') over $\mathbb{R}$, `formal' here meaning the lack of reference to operator domains. The first step to build up the notion of self-adjointness in such physical context is to make the declaration of the operator (or form) domain somewhat ``physically inescapable'' (which is at the very opposite to being a mere mathematical technicality). Section \ref{sec:domains} is devoted to this first goal.

As elementary as it is, this first part of our programme is aimed at stressing that the sole formal action of an explicit operator, emerging, for example, from first quantisation arguments, does not qualify it as an observable. Nor is it possible to dismiss the domain choice to a sort of ``automatic'' assignment, for instance identifying the domain as the largest or the smallest selection of vectors of the underlying Hilbert space the considered formal operator can be meaningfully applied to. Both are tacit and typical ``temptations'' among many physicists and this explains our emphasis on this point. Of course, such temptations are harmless for all those quantum observables represented by bounded linear operators on Hilbert space, but clearly first quantisation produces also unbounded operators in the Schr\"{o}dinger representation.

Besides, once one finally convinces oneself that quantum observables are a special class of linear operators on Hilbert space identified by the simultaneous declaration of their domain and their action on the vectors of such domain, it is natural to define self-adjointness and hermiticity as two similar but in general distinct notions, as well as to introduce the physically meaningful, auxiliary concept of closed operators. These are the topics of the final part of Section \ref{sec:domains}.

In Section \ref{sec:emergence_SA} we turn to the core of our programme, that is, the discussion of the main physical motivations that  translate mathematically into the requirement that quantum observables be self-adjoint, and not just hermitian operators on Hilbert space. Hermiticity, as said, stems from the physical need of real expectations and real eigenvalues, and now we want to present physically grounded reasons for the stronger demand of self-adjointness.

\small

 \begin{table}[t!]
 \begin{tabular}{|c|c|}
 \hline\hline
  \!\!\begin{tabular}{c}
   \textsf{\textbf{Quantum observables} are represented by linear operators $A$} \\
  \textsf{ on Hilbert space $\cH$. Physical relevance of hermiticity:} \\
   \textsf{an operator $A$ has only real expectations $\langle\psi,A\psi\rangle$} $\Leftrightarrow$ \textsf{$A$ is hermitian} 
  \end{tabular}\!\!  & \begin{tabular}{c} \textsf{Sect.~\ref{intro}} \\ \textsf{Sect.~\ref{sec:domains} (intro)} \\ \textsf{Eqn.}~\eqref{eq:defsa1}-\eqref{eq:defsa1ter} \end{tabular} \\
  \hline\hline
  \!\!\begin{tabular}{c}
   \textsf{first quantisation $\rightarrow$ yields operators on $L^2$-space like \eqref{eq:op-list}} \\ \textsf{they are \textbf{formally} hermitian}   
  \end{tabular}\!\! &  \begin{tabular}{c} \textsf{Sect.~\ref{intro}} \\ \textsf{Sect.~\ref{sec:domains} (intro)} \end{tabular} \\
  \hline
  \!\!\begin{tabular}{c|c}
   \textsf{the formal action $\psi\mapsto A\psi$} & \textsf{the formal expectation $\psi\mapsto\langle\psi,A\psi\rangle$} \\
   \textsf{of certain formally hermitian} & \textsf{of certain formally hermitian} \\
   \textsf{operators on $L^2$-space} & \textsf{operators on $L^2$-space} \\
   \textsf{is not applicable to all $\psi$'s:} & \textsf{is not applicable to all $\psi$'s:} \\
   \textsf{non-$L^2$ output may occur} & \textsf{infinite expectation may occur}
  \end{tabular}\!\! & \begin{tabular}{c} \textsf{Sect.~\ref{sec:op-form-domain}} \\ \textsf{Example \ref{ex:not-in-operator-domain}} \\ \textsf{Example \ref{ex:not-in-form-domain}} \end{tabular} \\
  \hline
  \!\!\begin{tabular}{c}
  \textsf{\textbf{operator domain} $\mathcal{D}(A)$ and \textbf{form domain} $\mathcal{D}[A]$ of certain} \\ \textsf{hermitian operators $A$ are proper subspaces of $\cH$}
  \end{tabular}\!\! &  \textsf{Sect.~\ref{sec:op-form-domain}} \\
  \hline
  \!\!\begin{tabular}{c}
  $\rightarrow$ \textsf{the reason is the \textbf{unboundedness} of $A$:} \\ \textsf{an everywhere-defined hermitian operator is necessarily bounded}
  \end{tabular}\!\! & \!\!\begin{tabular}{c} \textsf{Hellinger-Toeplitz} \\ \textsf{Theorem \ref{thm:hellingertoeplitz} } \\ \textsf{Example \ref{ex:unbdd-everywhere}} \end{tabular}\!\! \\
  \hline
    \!\!\begin{tabular}{c}
  \textsf{the \textbf{canonical commutation relation} $QP-PQ=\ii\hbar$ on $L^2(\mathbb{R}^d)$} \\ \textsf{can only be satisfied if at least one among $P,Q$ is unbounded}
  \end{tabular}\!\! & \!\!\begin{tabular}{c} \textsf{Example \ref{ex:Winter-Wielandt}} \\ \textsf{Example \ref{Popas-example}} \end{tabular}\!\! \\
  \hline
  \!\!\begin{tabular}{c}
  \textsf{is the association of $\mathcal{D}(A)$ and $\mathcal{D}[A]$ to a formal $A$ `automatic'? \textbf{no}:} \\
  \textsf{$\rightarrow$ choosing the domain \emph{maximally} is incompatible with hermiticity} \\
  \textsf{$\rightarrow$ there may be no maximal domain of hermiticity} \\
   \textsf{$\rightarrow$ no non-trivial \& unambiguous notion of minimal domain of hermiticity}
  \end{tabular}\!\! & \!\!\begin{tabular}{c} \\ \textsf{Sect.~\ref{sec:no-domain-maximally}} \\ \textsf{Sect.~\ref{sec:no-max-domain-hermiticity}}  \\ \textsf{Sect.~\ref{sec:nominimaldom}}  \end{tabular}\!\! \\
%   \hline
%   \!\!\begin{tabular}{c}
%        \textsf{neglecting at all observables' domain by introducing } \\
%        \textsf{\textbf{generalised vectors} and \textbf{infinite energies} yields inconsistencies}
%       \end{tabular}
%   \!\! & \textsf{Sect.~\ref{sec:infinities_ambiguous}} \\ 
  \hline
  \!\!\begin{tabular}{c}
  \textsf{conclusion: difference between formal action $A$ and operator $(A,\mathcal{D}(A))$;} \\ \textsf{$(A,\mathcal{D}_1)$ and $(A,\mathcal{D}_2)$ with $\mathcal{D}_1\neq\mathcal{D}_2$ are different observables}
   \end{tabular}\!\! & \textsf{Sect.~\ref{sec:ADA}} \\
   \hline
  \!\!\begin{tabular}{c}
  \textsf{additional conclusion: there is physics in the domain declaration;} \\ \textsf{role of boundary conditions}
   \end{tabular}\!\! & \textsf{Sect.~\ref{sec:bc}} \\
   \hline
  \textsf{only when $\mathcal{D}(A)$ is dense in $\cH$, is the \textbf{adjoint} $A^\dagger$ unambiguously defined} & \textsf{Sect.~\ref{sec:density-of-domain}} \\
   \hline
   \!\!\begin{tabular}{c}
   \textsf{once the non-triviality of the domain declaration is understood,} \\ \textsf{it finally makes sense to define \textbf{hermitian} vs \textbf{self-adjoint}:} \\ \textsf{identical notions for bounded $A$, in general distinct for unbounded $A$}
    \end{tabular}\!\! & \textsf{Sect.~\ref{sec:hermitian-selfadj}}  \\
   \hline
    \!\!\begin{tabular}{c}
   \textsf{self-adjoint operators may be unbounded (hence non-continuous),} \\ \textsf{but at least are all \textbf{closed operators};} \\ \textsf{instead, unbounded hermitian operators are not necessarily closed
} \end{tabular}\!\! & \textsf{Sect.~\ref{sec:closedoperators}}  \\
   \hline
    \!\!\begin{tabular}{c} 
   \textsf{\textbf{algebraic manipulation} of \textbf{unbounded} quantum observables} \\
   \textsf{is a touchy business: paradoxes and erroneous conclusions} \\
   \textsf{if \textbf{domain issues} are overlooked}
   \end{tabular}\!\! & \textsf{Sect.~\ref{sec:domains-touchy}} \\
   \hline\hline
 \end{tabular}
 \vspace{0.1cm}
 \caption{\label{tab:scheme1} Synoptic scheme of the main conceptual steps -- first part (Section \ref{sec:domains})} 
  \end{table}

 \begin{table}[t!]
 \begin{tabular}{|c|c|}
 \hline\hline
   \!\!\begin{tabular}{c} 
   \textsf{What physical requirements on quantum observables} \\ \textsf{ prescribe them to be self-adjoint and not merely hermitian?}
   \end{tabular}\!\! & \\
   \hline
   \!\!\begin{tabular}{c} 
   \textsf{the Schr\"{o}dinger equation $\ii\partial_t\psi(t)=H\psi(t)$, $\psi(0)=\psi_0\in\mathcal{D}(A)$} \\ \textsf{determines a \textbf{unique} solution $\psi(t)\in\mathcal{D}(H)$ that evolves } \\
   \textsf{\textbf{unitarily}, \textbf{strongly continuously (and with group property) in time}} \\
   \textsf{\emph{if and only if} the Hamiltonian $H$ is self-adjoint, not merely hermitian}
   \end{tabular}\!\! & \begin{tabular}{c} \textsf{Sect.~\ref{sec:selfadj-spectralthm}} \\ \textsf{Theorem \ref{thm:selfadj-SchrEq}} \end{tabular}\\
   \hline
      \!\!\begin{tabular}{c} 
   \textsf{self-adjointness of $H$ gives rise to \textbf{spectral theorem} / \textbf{functional calculus}} \\ \textsf{so as to build the \textbf{Schr\"{o}dinger propagator} $e^{-\ii t H}$:} \\
   \textsf{requiring a \textbf{dense of analytic vectors} $\psi$ with $ \sum_{n=0}^\infty\|H^n\psi\| t^n/n!<+\infty$} \\ % $e^{-\ii t H}\psi=\sum_{n=0}^\infty\frac{(-\ii t)^n}{n!} H^n \psi$
   \textsf{ necessarily makes a closed hermitian $H$ self-adjoint}
   \end{tabular}\!\! & \!\!\!\begin{tabular}{c} \textsf{Sect.~\ref{sec:analyticvectors}} \\ \textsf{Nelson's theorem:} \\ \textsf{Theorem \ref{thm:nelson}} \end{tabular} \!\!\!\!\\
   \hline
   \!\!\begin{tabular}{c}
   \textsf{an unphysical phenomenon:} \\
   \textsf{Schr\"{o}dinger dynamics $\psi(t)$ originating from a given initial datum $\psi_0$} \\
   \textsf{is non-unique in the lack of an explicit declaration of self-adjointness}
   \end{tabular}\!\! & \textsf{Sect.~\ref{sec:non-uniquedynamics}} \\
   \hline
   \!\!\begin{tabular}{c}
   \textsf{for a generic quantum observable $A$, self-adjointness (and not mere} \\
   \textsf{hermiticity) is imposed by analogy with the quantum Hamiltonian} 
   \end{tabular}\!\! & \textsf{Sect.~\ref{sec:SA-generic-obs}} \\
   \hline
   \!\!\begin{tabular}{c}
   \textsf{requiring a quantum observable, as a \underline{closed} hermitian operator $A$,} \\
   \textsf{to have an \textbf{orthonormal basis of eigenstates} makes $A$ self-adjoint} 
   \end{tabular}\!\! & \begin{tabular}{c} \textsf{Sect.~\ref{sec:OBN-estates}} \\ \textsf{Theorem \ref{thm:onbEV-sa}} \\ \textsf{Example \ref{ex:qao}} \end{tabular}\\
   \hline
   \!\!\begin{tabular}{c}
   \textsf{requiring certain observable expectations to behave as} \\
   \textsf{a \textbf{densely defined}, \textbf{lower semi-bounded}, \textbf{closed quadratic form}}  \\
   \textsf{forces the underlying linear operator $A$ associated with the form} \\
   \textsf{to be self-adjoint (and not merely hermitian)}
   \end{tabular}\!\! & \begin{tabular}{c} \textsf{Sect.~\ref{sec:closedsemibdd-sa}} \\ \textsf{Theorem \ref{thm:closed-sa}}  \end{tabular}\\
   \hline
   \!\!\begin{tabular}{c}
    \textsf{only self-adjointness of the observable $A$ ensures the consistent} \\
    \textsf{and non-ambiguous expansion} $\psi=\sum_n c_n\psi_n+\int c(\lambda)\psi_\lambda\ud\lambda$ $\forall\psi\in\cH$ \\
     \textsf{in terms of \textbf{eigenfunctions} and  \textbf{generalised eigenfunctions} of $A$}
   \end{tabular} & \textsf{Sect.~\ref{sec:generalized-EF}} \\
   \hline\hline
 \end{tabular}
 \vspace{0.1cm}
 \caption{\label{tab:scheme2} Synoptic scheme of the main conceptual steps -- second	 part (Section \ref{sec:emergence_SA})} 
  \end{table}

%      \hline
%    \!\!\begin{tabular}{c}
%    \textsf{only self-adjointness of the observable $A$} \\
%    \textsf{ensures the expansion $\psi=\sum_n c_n\psi_n+\int c(\lambda)\,\psi_\lambda\,\ud\lambda$}  \\
%    \textsf{in terms of eigenfunctions and generalised eigenfunctions of $A$}
%    \end{tabular}\!\! & \textsf{Sect.~\ref{sec:generalized-EF}}  \end{tabular}\\
  
  %\psi\;=\;\sum_n c_n\psi_n+\int c(\lambda)\,\psi_\lambda\,\ud\lambda
  
\normalsize

We follow a sort of hierarchical order in importance, starting with the most relevant observables: quantum Hamiltonians (the observables governing the evolution in time of the considered quantum systems). We thus discuss self-adjointness as that feature of quantum Hamiltonians that, unlike mere hermiticity, makes them the generators of the quantum dynamics with all the expected physical characteristics (unitarity, strong continuity in time, group property at different instants of time). 
%%%%%%%%%%%%%%%%%%%%%%%%%%%%%%%%%%%%%%%%%%%%%%%%%%%%
%%%%%%%%%%%%%%%%%%%%%%%%%%%%%%%%%%%%%%%%%%%%%%%%%%%%
%%%%%%%%%%%%%%%%%%%%%%%%%%%%%%%%%%%%%%%%%%%%%%%%%%%%
Next, we discuss how self-adjointness is dictated from the frequent and physically relevant circumstance where the operator domain contains an orthonormal basis of eigenvectors (an old line of reasoning that is already present in the first historical constructions of the mathematical structure of quantum mechanics). In addition, we examine the necessity of self-adjointness owing to the physical requirement that the expectations of certain relevant observables be uniformly bounded from below (like for the Hamiltonian of a stable quantum system) and behave as a lower semi-continuous quadratic form. Last, we outline the crucial relevance that self-adjointness (unlike mere hermiticity) has in allowing for the expansion of a generic state (vector in the Hilbert space) into ordinary and generalised eigenvectors of a quantum observable.

The main steps of our programme are visualised in the synoptic schemes of Tables \ref{tab:scheme1} and \ref{tab:scheme2}.

%Additional conclusion: there is physics in the declaration of the domain. Role of the boundary conditions

The pedagogical path developed throughout Sections \ref{sec:domains} and \ref{sec:emergence_SA} is formulated with an amount of mathematics (basics from functional analysis and operator theory) that lies presumably at the edge of the technical arsenal one is equipped with in the course of a first physical introduction to quantum mechanics at an undergraduate level -- but is surely part of the minimal background on mathematical methods for physics which one learns soon after  (we give for granted the notion of dense subspace in an infinite-dimensional Hilbert space, orthonormal vs algebraic basis, orthogonal complement, concrete $L^2$-spaces, and relevant subspaces such as the Schwartz functions, whereas we revisit the notion of adjoint operator and introduce standard Sobolev spaces ``operationally'' with only a tacit reference to weak derivatives and distributions). Most importantly, we intended all such mathematical machinery to emerge and be dealt with in very close connection with the physical reasoning that unfolds along the way.
%, so as to keep stringent physical grounds to accomodate quantum observables in the theory in the form of operators that are self-adjoint, and not merely hermitian. 
We do so also by presenting a progression of concrete examples that we believe can be instructively worked out in the course of the main line of reasoning.

Two additional Sections contain supplementary materials that we reckon to be equally instructive. In Section \ref{sec:selfadjproblem} we collected a list of the most representative categories of self-adjointness problems, solved or still under investigation in quantum mechanics: we refer here to the problem of rigorously proving self-adjointness for quantum observables of relevance in the applications, where the formal action of such observables is dictated by physical heuristics. The two-fold goal is to emphasise that this has been a non-trivial problem in the past, and it is still active for quantum models of recent theoretical of applied importance.

Last, we deferred to Section \ref{sec:math-proofs} the mathematical proofs of certain fundamental results that translate various physical requirements into the notion of self-adjointness (Theorems \ref{thm:selfadj-SchrEq}, \ref{thm:onbEV-sa}, \ref{thm:closed-sa}). Such proofs, as classical as they are by now, are somewhat more elaborated than the rest of the discussion in Sections \ref{sec:domains} and \ref{sec:emergence_SA}, and we find convenient not to interrupt the main reasoning therein. Yet, we believe that it is beneficial to include such proofs as part of our pedagogical path, and for this reason we did not merely made reference to the literature (which would require a multiplicity of separate facts to be cited from different contexts), but we instead assembled them in a form that makes them accessible to any mathematically educated physical readership.

\section{First part: declaring the domain is inescapable}\label{sec:domains}

We start our discussion from the typical observables in the Schr\"{o}dinger representation for a one-dimensional (spinless) quantum particle, namely formal operators of the type \eqref{eq:op-list} acting in the Hilbert space $L^2(\Omega)$, where $\Omega\subset\mathbb{R}$ is in practice a finite or infinite interval, or union of intervals, or the whole real line.

This is the concrete playground for a more abstract setting in which $\cH$ is a \emph{complex} Hilbert space (with the convention, throughout this work, that the associated scalar product $\langle\cdot,\cdot\rangle$ is anti-linear in the first entry and linear in the second) and $A$ is a linear operator acting on $\cH$, in particular an operator associated with a generic quantum observable. For the choices \eqref{eq:op-list} $A$ is a (pseudo-)differential operator of at most second order on $L^2(\Omega)$.

We should rather write `formal operator' $A$ as long as its domain remains unspecified, meaning that for the time being we only refer to the recipe $\psi\mapsto A\psi$ that produces the output $A\psi$ given the input $\psi$.

As commented already, \emph{hermiticity} of $A$ is a fairly understandable feature at any however elementary level of physical discussion of the mathematical framework of quantum mechanics, as hermitian operators are the sole class of operators in Hilbert space with \emph{real} expectations (besides, a hermitian operator only admits real eigenvalues). When $A$ is everywhere defined and bounded on the considered Hilbert space $\cH$, the symmetry property
\begin{equation}\label{eq:symmetry}
  \langle \psi,A\psi\rangle\;=\;\langle A\psi,\psi\rangle
\end{equation}
is spelled over every vector $\psi\in\cH$. Owing to the unboundedness of operators like \eqref{eq:op-list}, for them the ``generic practitioner'' of quantum mechanics only checks \eqref{eq:symmetry} on a tacitly meaningful linear subspace of functions $\psi\in L^2(\mathbb{R})$ that are suitably smooth and vanish sufficiently fast at infinity, so as to make each side of \eqref{eq:symmetry} well defined: in this case the identity \eqref{eq:symmetry} follows from integration by parts. (One customarily says that a differential operator like \eqref{eq:op-list} is `formally self-adjoint' \cite[Sect.~4.1]{Grubb-DistributionsAndOperators-2009}.)

In this Section we intend to focus on the association of the formal operator action $A$ with an operator domain of states it acts on, which is to be declared \emph{simultaneously} with the declaration of the formal action of $A$. We do not want to merely associate a domain with $A$ as a part of a mathematical definition (with the risk of downgrading it to a technicality): we rather want the ``need for a domain'' to emerge as a non-trivial, non-automatic, physically meaningful declaration of admissible states, that no practitioner of quantum mechanics can escape.

Of course all this is fairly basic in functional analysis and operator theory, but let us recall once again that we are having in mind a conceptual path where the mathematical formalisation emerges in a physical context, as is the case for those ideal readers we referred to in the introduction.

\subsection{Operator domain and form domain}\label{sec:op-form-domain}~

This is standard mathematical language, essentially digestible at any level, so let us introduce it once for all. Already the first examples we provide should convince that such language is not void.

By `domain' one means a suitable linear subspace of $\cH$ on which the action of $A$ is meaningful -- linearity of the domain is necessary for consistency with the linearity of $A$ and with the superposition principle in quantum mechanics.

In particular, with `operator domain' \emph{associated with the formal action $A$}, one refers to a (linear) subspace $\mathcal{D}(A)\subset\cH$ whose elements $\psi$ satisfy $A\psi\in\cH$, namely the output of the formal action $A$ applied to each such $\psi$ is a vector in $\cH$. This notion is practically irrelevant when the formal operator $A$ is \emph{bounded} on $\cH$, namely when $\|A\|_{\mathrm{op}}<+\infty$, where
\begin{equation}\label{eq:boundedness}
 \|A\|_{\mathrm{op}}\;:=\;\sup_{\substack{\psi\in\cH \\ \|\psi\|\neq 0}}\frac{\|A\psi\|}{\|\psi\|}\,,
\end{equation}
for then the formal action of $A$ on \emph{any} $\psi\in\cH$ yields $A\psi\in\cH$. It is under such tacit assumption of boundedness that one discusses quantum observables in a first physical introduction to quantum mechanics, like in Dirac's celebrated Principles of Quantum Mechanics: 
%``A linear operator is considered to be completely defined when the result of its application to every ket vector is given'' \cite[Sect.~7]{Dirac-PrinciplesQM}. 

\begin{quote} 
\centering 
``A linear operator is considered to be completely defined when the result of its application to every ket vector is given'' \cite[Sect.~7]{Dirac-PrinciplesQM}.
\end{quote}

It is straightforward, on the other hand, to produce examples where the formal action of $A$ does not map $\psi$ into a vector in $\cH$.

\begin{example}\label{ex:not-in-operator-domain}
With $\cH=L^2(\mathbb{R})$, let $A=-\ii\frac{\ud}{\ud x}$ or $A=-\frac{\ud^2}{\ud x^2}$ and $\psi(x)=|x|^{-1/4}e^{-x^2}$. Then $\psi\in L^2(\mathbb{R})$ but $\psi'\notin L^2(\mathbb{R})$ and $\psi''\notin L^2(\mathbb{R})$. Analogously, let $A$ be the multiplication by $x$ and $\psi(x)=(1+x^2)^{-3/4}$: then $\psi\in L^2(\mathbb{R})$ but $x\psi\notin L^2(\mathbb{R})$. For such observables, the admissible $\psi$'s on which to evaluate $A$ constitute a necessarily proper subspace of the Hilbert space. 
\end{example}

As for the `form domain' \emph{associated with the formal action $A$} (also called quadratic form or energy form of $A$, depending on the context), this is another linear subspace, for which the notation is now $\mathcal{D}[A]$, consisting of vectors on which the expectation of $A$ can be computed and is finite. More precisely, the actual quantity one is meant to evaluate here is a \emph{generalisation} of the ordinary expectation $\langle\psi, A\psi\rangle$ (namely the scalar product of two vectors in $\cH$, if $\psi\in\mathcal{D}(A)$), and is denoted for this reason with the new symbol $A[\psi]$ (in Sect.~\ref{sec:closedsemibdd-sa} we will also write $\mathcal{E}_A[\psi]$ to emphasise its meaning of ``energy''). Such generalised quantity $A[\psi]$ is defined by redistributing the formal action of $A$ in a quadratic sense (see examples in a moment) with the prescription that for vectors in the \emph{operator} domain one must have
\begin{equation}\label{eq:form-op-id}
 A[\psi]\;=\;\langle\psi, A\psi\rangle\qquad (\psi\in\mathcal{D}(A))\,.
\end{equation}
Thus, for instance, when $\cH=L^2(\mathbb{R})$ and $A=-\frac{\ud^2}{\ud x^2}$, strictly speaking
\[
 \langle\psi,A\psi\rangle\;=\;-\langle\psi,\psi''\rangle\;=\;-\int_\mathbb{R}\overline{\psi(x)}\psi''(x)\,\ud x\,;
\]
therefore, if the operator domain $\mathcal{D}(A)$ consists, say, of suitably regular and fast decreasing functions $\psi$, then integration by parts yields
\[
 \langle\psi,A\psi\rangle\;=\;-\int_\mathbb{R}\overline{\psi(x)}\psi''(x)\,\ud x\;=\;\int_\mathbb{R}|\psi'(x)|^2\,\ud x\qquad\forall\psi\in\mathcal{D}(A)\,.
\]
This leads one to define in this case the energy expectation of $A$ as
\[
 A[\psi]\;:=\;\int_\mathbb{R}|\psi'(x)|^2\,\ud x\,,
\]
with the second derivative originally hitting one $\psi$ only now redistributed as first derivative on both $\overline{\psi}$ and $\psi$. The latter is the expression every physicists knows well for the free kinetic energy of the state $\psi$. Observe that the latter identity \emph{defines} $A[\psi]$ for the considered formal action $A=-\frac{\ud^2}{\ud x^2}$ and satisfies \eqref{eq:form-op-id} above, \emph{but} \eqref{eq:form-op-id} is \emph{not} the definition of $A[\psi]$, it only equates $A[\psi]$ to $\langle\psi,A\psi\rangle$ for those special $\psi$'s for which $A\psi$ is a vector in the Hilbert space.

This way one declares two subspaces $\mathcal{D}(A)$ and $\mathcal{D}[A]$ for a given formal action $A$ on $\cH$. Owing to the constraint \eqref{eq:form-op-id}, obviously $\mathcal{D}(A)\subset \mathcal{D}[A]$. In general such two domains do not coincide, nor is $\mathcal{D}[A]$ in general the whole $\cH$. In particular, there may be states $\psi$ for which $A[\psi]$ is finite, and therefore $\psi\in\mathcal{D}[A]$, but $\langle\psi,A\psi\rangle$ is infinite, and therefore $\psi\notin\mathcal{D}(A)$.

\begin{example}\label{ex:not-in-form-domain}
With respect to $\cH=L^2(\mathbb{R})$, when $A=-\frac{\ud^2}{\ud x^2}$ and $\psi(x)=|x|^{3/2}e^{-x^2}$ one has $\psi\in L^2(\mathbb{R})$ and $\psi''\notin L^2(\mathbb{R})$, meaning that $\psi$ cannot be ascribed to $\mathcal{D}(A)$, yet $\psi'\in L^2(\mathbb{R})$ and therefore $\psi$ can be ascribed to $\mathcal{D}[A]$ (in short: on $\psi$ one cannot evaluate the kinetic energy operator, but can evaluate its expectation). Instead, with $\psi(x)=|x|^{1/2}e^{-x^2}$ one has $\psi\in L^2(\mathbb{R})$ and $\psi'\notin L^2(\mathbb{R})$, meaning that one cannot evaluate the expectation of $A$ on $\psi$, thus $\psi$ cannot be ascribed to $\mathcal{D}[A]$. Analogously, when $A$=multiplication by $x$ and $\psi(x)=(1+x^2)^{-3/4}$, one has $\psi\in L^2(\mathbb{R})$ and $x\psi\notin L^2(\mathbb{R})$ (one cannot evaluate the position operator on $\psi$), yet $\int_\mathbb{R}x|\psi(x)|^2\ud x$ is finite (one can evaluate on $\psi$ the expectation of the position operator).
\end{example}

%As is clear from the above informal definitions, the subspace $\mathcal{D}(A)$ is more natural when abstractly considering $A$ as a linear mapping in $\cH$, whereas the subspace $\mathcal{D}[A]$ is more natural when one wants to qualify the class of states on which the expectation of $A$ is evaluated.

\subsection{$\mathcal{D}(A)$ and $\mathcal{D}[A]$ are in general proper subspaces. Connection with unboundedness - I}~

%We have thus recognised the rather evident fact that for certain observables $A$ not on all vectors of $\cH$ can one evaluate the action of $A$ or the expectation of $A$. In other words, $A$ in general only acts on an operator domain $\mathcal{D}(A)$ and a form domain $\mathcal{D}[A]$ that are \emph{proper} subspaces of $\cH$.

Examples \ref{ex:not-in-operator-domain}-\ref{ex:not-in-form-domain} show the rather evident fact that for certain observables $A$ their formal action or expectation cannot make sense on all the vectors of $\cH$, that is,
operator domain $\mathcal{D}(A)$ and form domain $\mathcal{D}[A]$ associated with the formal action of $A$ are only \emph{proper} and \emph{distinct} subspaces of $\cH$. Such examples involved observables from first quantisation which are \emph{unbounded} on $L^2(\mathbb{R})$, that is, whose operator norm \eqref{eq:boundedness} is infinite. This does not show, though, that it is precisely unboundedness to prevent $A$ to be everywhere defined.
%-- we merely discussed circumstances in which $\psi$ belongs to $L^2(\mathbb{R})$, but $A\psi$ does not or $A[\psi]$ is not finite. 

Such a point surely deserves being highlighted at this stage of our proposed pedagogical path: 
\begin{center}
\begin{tabular}{c}
 \emph{it is the fact that observables are \emph{hermitian} that makes for them} \\ 
 \emph{incompatible to be simultaneously unbounded and everywhere defined}.
\end{tabular}
\end{center}
One has indeed the following.

\begin{theorem}[The Hellinger-Toeplitz theorem]\label{thm:hellingertoeplitz}
Let $A$ be an everywhere defined linear operator on a Hilbert space $\cH$ with $\langle \psi,A\phi\rangle=\langle A\psi,\phi\rangle$ for all $\psi$ and $\phi$ in $\cH$. Then $A$ is bounded. 
\end{theorem}

The Hellinger-Toeplitz theorem is a standard consequence of the closed graph theorem, which is in turn a consequence of the completeness of $\cH$ via the Baire category theorem (see, e.g., \cite[Sect.~III.5]{rs1}): thus, although the steps for its proof pertain an abstract level that may divert one from a more physical reasoning, it should at least be stressed that we are encountering here an \emph{effect of completeness}. Completeness is indeed a feature of Hilbert spaces (among other topological structures) that is apparently innocent when its mathematical definition is laid down in class, but whose actual relevance in the conceptual structure of quantum mechanics is not immediate to spot.
%-- we shall come on this point in Subsect.~**** %\ref{sec:density-of-domain}.

Besides, should one not want to enter the proof of the Hellinger-Toeplitz theorem, it would be instructive to show at least a `concrete' example of an everywhere defined and unbounded operator $A$ on Hilbert space and to check for it the lack of hermiticity, thus making $A$ unsuited to represent a quantum observable. In fact, an example of that sort cannot be so ``concrete'', in that it requires the axiom of choice; nevertheless it may result useful also for a physical audience.

\begin{example}\label{ex:unbdd-everywhere}
 Let $\cH$ be an infinite-dimensional separable Hilbert space (meaning that $\cH$ admits a countably-infinite \emph{orthonormal basis}) like for example $\cH=L^2(\mathbb{R})$, and let $\mathcal{B}$ be an \emph{algebraic basis} of $\cH$, with all elements non-restrictively normalised to 1. Observe that the existence of $\mathcal{B}$ requires the axiom of choice and that $\mathcal{B}$ is necessarily uncountable. Moreover, let $\mathcal{B}'=\{\psi_n|n\in\mathbb{N}\}$ be a countable subset of $\mathcal{B}$ and let $\phi_\circ\in\cH$ with $\|\phi_\circ\|=1$. Define the linear operator $A$ by linear extension of
 \[
 A\psi\;:=\;
 \begin{cases}
  n\phi_\circ & \textrm{if $\psi=\psi_n$ for some $n$} \\
  \; 0 & \textrm{if }\,\psi\in \mathcal{B}\setminus\mathcal{B}'\,.
 \end{cases}
 \]
 By construction $A$ is everywhere defined. It is also unbounded, for $\|A\psi_n\|=n$. On the other hand, $\langle\psi_n,A\psi_n\rangle=n\langle\psi_n,\phi_\circ\rangle$, whereas $\langle A\psi_n,\psi_n\rangle=n\langle\phi_\circ,\psi_n\rangle$. Now it is easy to choose $\phi_0$ and $\psi_n$ so as to conclude that $A$ is not hermitian. 
\end{example}

\subsection{Connection with unboundedness - II}~

What argued so far should already be enough to explain the fact of life that most operators of relevance in quantum mechanics are unbounded and hence, owing to the additional requirement of hermiticity, cannot be defined everywhere on the underlying Hilbert space (Hellinger-Toeplitz theorem).

It is helpful at this stage to observe that from another perspective unboundedness is also dictated by the standard formulation of the canonical commutation relation ``$QP-PQ=\ii\hbar$'', with $Q$ and $P$ acting on $L^2(\mathbb{R})$ respectively as multiplication by $x$ and $-\ii\hbar\frac{\ud}{\ud x}$.

\begin{example}[Winter-Wielandt]\label{ex:Winter-Wielandt}
 If two everywhere defined and bounded operators $Q$ and $P$ on $\cH$ satisfied $QP-PQ=\ii\mathbbm{1}$, then $Q^2P-PQ^2=Q(PQ+\ii\mathbbm{1})-PQ^2=2\ii Q$, and inductively $Q^nP-PQ^n=\ii n Q^{n-1}$, whence
 \[
  n\|Q^{n-1}\|_{\mathrm{op}}\;\leqslant\; 2 \,\|Q^{n-1}\|_{\mathrm{op}}\|Q\|_{\mathrm{op}}\|P\|_{\mathrm{op}}\qquad \forall n\in\mathbb{N}\,.
 \]
If $\|Q^{n-1}\|_{\mathrm{op}}=0$ for some $n$, then solving the above hierarchy backwards would yield $Q=\mathbbm{O}$, which is incompatible with $QP-PQ=\ii\mathbbm{1}$. Then necessarily it is always $\|Q^{n-1}\|_{\mathrm{op}}>0$, whence $\|Q\|_{\mathrm{op}}\|P\|_{\mathrm{op}}\geqslant\frac{1}{2}n$. As $n\in\mathbb{N}$ is arbitrary, this shows that at least one among $Q$ and $P$ cannot be bounded, thus contradicting the assumption. This fact was noticed first by Wintner \cite{Wintner-1947-pq-qp} in 1947, who gave a proof based on the spectra of $PQ$ and $QP$, and re-proved soon after by Wielandt \cite{Wielandt-1949-pq-qp} in 1949 with the algebraic argument used here. 
\end{example}

Interestingly enough, the emergence of unboundedness described in the last example can be made quantitative.

\begin{example}[Popa]\label{Popas-example}
 Assume that on a given Hilbert space $\cH$ two everywhere defined and bounded operators $Q$ and $P$ satisfy 
 \[
  \big\|[Q,P]-\ii\mathbbm{1}\big\|_{\mathrm{op}}\;\leqslant\;\varepsilon
 \]
 for some $\varepsilon>0$. Then
 \[
  \|Q\|_{\mathrm{op}}\|P\|_{\mathrm{op}}\;\geqslant\;\frac{1}{2}\,\log\frac{1}{\varepsilon}\,.
 \]
 This quantifies the tendency of $Q$ or $P$ to have larger and larger norms in terms of the norm of the displacement between their commutator and the identity. This fact too is rather easy to see (we show its proof in Section \ref{sec:math-proofs}) and was observed by Popa \cite{Popa-1982-XD-DX} in 1981. The bound $\|Q\|_{\mathrm{op}}\|P\|_{\mathrm{op}}\geqslant\frac{1}{2}\log\frac{1}{\varepsilon}$ is expected to be optimal: examples of everywhere defined and bounded operators $Q$ and $P$ with $\big\|[Q,P]-\ii\mathbbm{1}\big\|_{\mathrm{op}}\leqslant\varepsilon$ and $\|Q\|_{\mathrm{op}}\|P\|_{\mathrm{op}}=O(\varepsilon^{-2})$ were constructed in \cite{Popa-1982-XD-DX}, and recently Tao \cite{Tao-2019-commutators-identity} produced examples of $Q,P$ with the same assumption and with $\|Q\|_{\mathrm{op}}\|P\|_{\mathrm{op}}=O(\log^5\frac{1}{\varepsilon})$. 
 \end{example}

It is true that at a later, and more fundamental stage, one reformulates Heisenberg's canonical commutation relation in the form of Weyl's commutation relation \cite[Section 3.1]{Strocchi-MathQM}, which is an identity between bounded operators and hence immune from domain issues. So, the argument of this Subsection alone can be regarded as just technical in a sense. Yet, it was instructive to present it as an additional point in the general perspective of the emergence of unbounded quantum observables.

\subsection{Declaring $\mathcal{D}(A)$ or $\mathcal{D}[A]$ maximally is incompatible with hermiticity}\label{sec:no-domain-maximally}~

After realising that the formal action $A$ of an operator representing a quantum observable may be only applicable to a proper subspace of states of the Hilbert space, \emph{one could legitimately suspect that there is an ``automatic'' way to declare the subspaces $\mathcal{D}(A)$ and $\mathcal{D}[A]$, with physical unambiguous content, once the formal action of $A$ is given}.

To begin with, one might decide to always associate with the formal operator $A$ a domain $\mathcal{D}(A)$ (respectively, $\mathcal{D}[A]$) consisting of the \emph{largest} possible subspace of vectors of $\cH$ for which $A\psi$ is still a vector in $\cH$ (respectively, for which the expectation of $A$ on $\psi$ is finite). It is easy to realise, though, that this prescription is in general incompatible with the hermiticity of $A$.

\begin{example}\label{eq:maximal-looses-hermiticity}
 Let $\cH=L^2(0,1)$ and $A=-\frac{\ud^2}{\ud x^2}$ (kinetic energy operator for a quantum particle in a box). 
 \begin{itemize}
  \item[(i)] The above-mentioned maximal choice for the operator domain of $A$ is clearly the linear subspace of all the $L^2$-functions on the interval $(0,1)$ whose second derivative is still square-integrable (understanding derivatives in the \emph{weak} sense): this is a classical functional space, the Sobolev space of second order
 \[
  \mathscr{H}^2(0,1)\;=\;\{\psi\in L^2(0,1)\,|\,\psi''\in L^2(0,1)\}\,.
 \]
 However, setting $\mathcal{D}(A)=\mathscr{H}^2(0,1)$ breaks the required hermiticity of $A$: for instance, with $\psi(x)=6x^2+(\ii-2)$, one has $\psi\in \mathscr{H}^2(0,1)$, but $\langle\psi,A\psi\rangle=12\,\ii\notin\mathbb{R}$.
 \item[(ii)] Analogously, the maximal choice for the form domain of $A$ is
 \[
  \mathscr{H}^1(0,1)\;=\;\{\psi\in L^2(0,1)\,|\,\psi'\in L^2(0,1)\}\,,
 \]
 and the same $\psi$ from (i) shows the loss of hermiticity.
 \end{itemize}
 \end{example}

\subsection{There may be no maximal domain of hermiticity}\label{sec:no-max-domain-hermiticity}~

As just seen, declaring $\mathcal{D}(A)$ `maximally' is in general incompatible with hermiticity. As a remedy for such an obstruction, \emph{still inquiring whether the formal action $A$ can be equipped with a canonically chosen domain}, one might decide to declare $\mathcal{D}(A)$ in a `conditionally maximal' sense as the largest possible subspace of $\cH$ on which $\langle\psi,A\phi\rangle=\langle A\psi,\phi\rangle$ for all $\psi$ and $\phi$ in the domain (or analogously for $\mathcal{D}[A]$).

However, the subspaces of $\cH$ on which the hermiticity of $A$ is guaranteed in general are not ordered by inclusion, and therefore there is no such largest of them.

\begin{example}\label{ex:momentum-box}
 Let $\cH=L^2(0,1)$ and $A=-\ii\frac{\ud}{\ud x}$ (momentum operator for a quantum particle in a box). For each $\theta\in[0,2\pi)$ consider the subspace
 \[
  \mathcal{D}_\theta\;:=\;\{\psi\in \mathscr{H}^1(0,1)\,|\,\psi(1)=e^{\ii\theta}\psi(0)\}\,,
 \]
 where $\mathscr{H}^1(0,1)$ is the Sobolev space of first order
 \[
  \mathscr{H}^1(0,1)\;=\;\{\psi\in L^2(0,1)\,|\,\psi'\in L^2(0,1)\}\,.
 \]
 Integration by parts shows that for each fixed $\theta$ one has
 \[
  \langle \psi, A\phi\rangle\;=\;\langle A\psi,\phi\rangle\qquad\forall\psi,\phi\in\mathcal{D}_\theta\,,
 \]
 therefore each subspace $\mathcal{D}_\theta$ is a domain of hermiticity for the formal operator $A$. However the $\mathcal{D}_\theta$'s are not ordered by inclusion and one cannot speak of the largest subspace of $L^2(0,1)$ on which the formal momentum operator is hermitian. 
\end{example}

\subsection{There is no non-trivial notion of minimal domain of hermiticity}\label{sec:nominimaldom}~

One more conceivable possibility for a canonical identification of the domain of a formal operator representing a quantum observable is to declare $\mathcal{D}(A)$ as the \emph{smallest} possible subspace of $\cH$ of hermiticity for $A$.

Yet, one soon sees that this does not work either. First of all, here ``smallest'' must be interpreted in a meaningful sense, so that the domain contains sufficiently many states to describe the quantum system under consideration, for instance a dense domain (for otherwise there is always the subspace $\{0\}$, consisting of only the zero vector, which is obviously a domain of hermiticity!). But ``smallest'' should also be understood unambiguously, if one insists on the case that assigning a domain to the formal action of a quantum observable is only an automatic, technical issue.

Now, it is easy to convince oneself that such a notion of minimal domain of hermiticity cannot exist. On the one hand, there is an inevitable arbitrariness in picking a domain among various candidate dense subspaces, when they are nested one into the other (Example \ref{ex-position-cInf0-S}). Besides, selecting a minimal domain by taking the intersection of all meaningful domains of hermiticity may well lead to a trivial domain, hence physically not informative (Example \ref{ex:position-op-schwartz-simple}).

\begin{example}\label{ex-position-cInf0-S}
 Let $\cH=L^2(\mathbb{R})$ and let $A$ be the multiplication by $x$. Consider the two subspaces
 \[
  \begin{split}
   \mathcal{D}_1\;&:=\;C^\infty_0(\mathbb{R})\;\equiv\;
   \left\{
   \begin{array}{c}
    \textrm{$\mathbb{R}\to\mathbb{C}$ functions that are} \\
    \textrm{infinitely many times differentiable,} \\
    \textrm{and whose support is a compact subset of $\mathbb{R}$} 
   \end{array}
   \right\}\,, \\
   \mathcal{D}_2\;&:=\;\mathcal{S}(\mathbb{R})\;\equiv\;
   \left\{\!
   \begin{array}{c}
    \textrm{`Schwartz functions', i.e.,} \\
    \textrm{infinitely differentiable $\mathbb{R}\to\mathbb{C}$ functions} \\
    \textrm{vanishing at infinity, together with all derivatives,} \\
    \textrm{faster than any polynomial}
   \end{array}\!\!
   \right\}\,.
  \end{split}
 \]
 As a matter of fact, both $\mathcal{D}_1$ and $\mathcal{D}_2$ are dense in $L^2(\mathbb{R})$, both are invariant under multiplication by $x$, that is, $A\mathcal{D}_1\subset\mathcal{D}_1$ and $A\mathcal{D}_2\subset\mathcal{D}_2$, and moreover $\langle\psi,A\phi\rangle=\langle A\psi,\phi\rangle$ for every pair $\psi,\phi\in\mathcal{D}_1$, and also every pair $\psi,\phi\in\mathcal{D}_2$. Thus, both $\mathcal{D}_1$ and $\mathcal{D}_2$ are legitimate dense domains of hermiticity for $A$. In fact, $\mathcal{D}_1\varsubsetneq\mathcal{D}_2$: which of the two is to be picked as ``\emph{the}'' domain of the position observable?
\end{example}

\begin{example}\label{ex:position-op-schwartz-simple}
 Let $\cH=L^2(\mathbb{R})$ and let $A$ be the multiplication by $x$. Consider the two subspaces
 \[
  \begin{split}
   \mathcal{D}_1\;:&=\;\{\textrm{Schwartz functions $\mathbb{R}\to\mathbb{C}$}\}\,, \\
   \mathcal{D}_2\;:&=\;\{\textrm{really simple functions $\mathbb{R}\to\mathbb{C}$}\}\,,
  \end{split}
 \]
 where `really simple functions' are those step-wise constant functions given by finite linear combinations of characteristic functions of intervals of finite length \cite[Sect.~1.17-1.18]{Lieb-Loss-Analysis}. Both $\mathcal{D}_1$ and $\mathcal{D}_2$ are dense in $L^2(\mathbb{R})$. Clearly $x\psi\in L^2(\mathbb{R})$ and $\langle\psi,A\phi\rangle=\langle A\psi,\phi\rangle$ for every pair $\psi,\phi\in\mathcal{D}_1$, and also every pair $\psi,\phi\in\mathcal{D}_2$. Thus, both $\mathcal{D}_1$ and $\mathcal{D}_2$ are dense domains of hermiticity for the quantum position operator. However, $\mathcal{D}_1\cap \mathcal{D}_2=\{0\}$: the smallest subspace of hermiticity for $A$ is the trivial subspace. 
 %The same happens in the setting of Example \ref{ex:momentum-box} above: $\cap_{\theta}\mathcal{D}_\theta=\{0\}$. 
\end{example}

\subsection{Temporary conclusion: Attached to the formal action of $A$ is a non-trivial domain $\mathcal{D}(A)$ or $\mathcal{D}[A]$}\label{sec:ADA}~

The simple reasonings of Subsect.~\ref{sec:op-form-domain}-\ref{sec:nominimaldom} show that the formal action of the operator corresponding to a quantum observable only makes sense in association with a domain of admissible states for the operator's action or the operator's expectation: such a domain is in general -- namely for hermitian unbounded operators -- a proper subspace of the underlying Hilbert space, and there is no canonical way of declaring it on the sole basis of the formal action.

This means that for the \emph{same} formal action $A$ and any two admissible domains $\mathcal{D}_1$ and $\mathcal{D}_2$ of hermiticity for $A$ with $\mathcal{D}_1\neq\mathcal{D}_2$, the pairs $(A,\mathcal{D}_1)$ and $(A,\mathcal{D}_2)$ are candidate to represent \emph{distinct} quantum observables.

As is customary, one then tacitly adopts the convention that when referring to the `operator $A$' one is indeed declaring \emph{simultaneously} the formal action of $A$ and its (operator or form) domain.

Once again all this is standard from the perspective of abstract functional analysis and operator theory, but we believe that the simple considerations developed so far fit well the typical physical introduction to the mathematical framework of quantum mechanics: an observable is represented by a (so far) hermitian operator which comes with its own domain.

Besides, it is only after having digested that an unbounded hermitian operator on an infinite-dimensional Hilbert space has a non-trivial domain that cannot be inferred by the sole formal action of the operator, that one can understand the notion of self-adjointness in comparison with the weaker notion of mere hermiticity (see Subsect.~\ref{sec:hermitian-selfadj} below).

\subsection{Additional conclusion: there is physics in the declaration of the domain. Interpretation of the boundary conditions.}\label{sec:bc}~

In retrospect, the previous reasonings provide the ground for this additional and fundamental conclusion: there is a physical content also in the declaration of the domain for the formal action of a quantum observable. Indeed, from the \emph{same} formal action $A$ one can well have \emph{distinct} observables $(A,\mathcal{D}_1)$ and $(A,\mathcal{D}_2)$, whose difference lies in the distinct domains $\mathcal{D}_1$ and $\mathcal{D}_2$.

This too is fairly standard from a more advanced perspective. We place it at this stage of our proposed pedagogical path to complement the tacit point of view of typical physical introductions to quantum mechanics, where the physical meaning of a quantum observable is solely attributed to the formal action of the operator that represents the observable -- think of the commonplace discussion on the momentum observable $-\ii\frac{\ud}{\ud x}$ when its \emph{action} is interpreted as the generator of spatial translations \cite[Sect.~22]{Dirac-PrinciplesQM}, \cite[\S 15]{Landau-Lifshitz-3}, \cite[Sect.~1.6]{sakurai_napolitano_2017}, \cite[Sect.~3.5]{weinberg_2015}, with no reference to the operator domain.

In practice, in the first quantisation setting, distinct domains of hermiticity (and eventually of self-adjointness, which is the true physical requirement) differ by the conditions assigned at the boundaries of the spatial region where the considered quantum system lives in.

Quite often each boundary condition has a transparent physical interpretation. In this regard, probably the most typical example used in class is the following.

\begin{example}\label{ex:bc}
 For a free quantum particle confined in the interval $[0,1]$, thus with Hilbert space $\cH=L^2(0,1)$, the `free energy' Hamiltonian has the formal action $H=-\frac{\ud^2}{\ud x^2}$. The subspaces
 \[
  \begin{split}
   \mathcal{D}_D\;&:=\;\left\{
   \psi\in L^2(0,1)\left|
   \begin{array}{c}
    \psi''\in L^2(0,1), \\
    \psi(0)=0=\psi(1)
   \end{array}\!\!
   \right.\right\}, \\
    \mathcal{D}_P\;&:=\;\left\{
   \psi\in L^2(0,1)\left|
   \begin{array}{c}
    \psi''\in L^2(0,1), \\
    \psi(0)=\psi(1), \; \psi'(0)=\psi'(1)
   \end{array}\!\!
   \right.\right\}
  \end{split}
 \]
 are two dense domains of hermiticity for $H$ (in fact, of \emph{self-adjointness}, as we shall see later on): the symmetry property \eqref{eq:symmetry} for $H$ is ensured precisely by the fact that both the \emph{Dirichlet} and the \emph{periodic} boundary conditions above make the boundary terms in the integration by parts vanish. The pairs $(H,\mathcal{D}_D)$ and $(H,\mathcal{D}_P)$ represent two distinct quantum observables: the free Hamiltonian with Dirichlet boundary conditions describes the physics where the walls at $x=0$ and $x=1$ \emph{repel} the particle away, whereas the periodic boundary conditions encode a kind of \emph{attractive} walls, as the particle is allowed to have a non-vanishing wave-function in their vicinity. The two operators display other substantial differences such as in their spectra (see Subsect.~\ref{sec:non-uniquedynamics} below). 
\end{example}

The role and physical interpretation of different boundary conditions for the same formal Schr\"{o}dinger operator can be ubiquitously spotted for more involved and realistic quantum mechanical Hamiltonians and quantum observables in general -- the physics-oriented overview \cite{Araujo-Coutinho-Perez-2004} present a series of instructive examples.  This is the case, technically speaking, whenever the formal operator under consideration admits a multiplicity of distinct self-adjoint realisations, that in practice are identified by distinct boundary conditions. Such perspective will be discussed again in Subsect.~\ref{sec:non-uniquedynamics}.

\subsection{On the density of the domain}\label{sec:density-of-domain}~

Let us collect some side remarks
%, that can be made already at this intermediate stage of our proposed conceptual path, 
concerning a feature of the domain of a quantum observable: its density in the underlying Hilbert space.

The preceding arguments do not decide whether the domain of a quantum observable need be actually dense. In retrospect, self-adjointness does require density of the domain, and regarding for instance position and momentum operators as \emph{generators} of the Weyl $C^*$-algebra, or a quantum Hamiltonian as the generator of the unitary evolution, Stone's theorem then implies that such generators are all densely defined. Yet, it is not immediate to kick density in at this level and by means of physically grounded motivation (apart from the generic request to have ``sufficiently many'' physical states).

In fact, the very definition of hermitian operator does not require its domain to be a dense subspace. Thus, the multiplication by $x$ on $L^2(\mathbb{R})$ is obviously hermitian both on the subspace $\mathcal{S}(\mathbb{R})$ of Schwartz functions of the real line, or on the subspace of $L^2$-functions whose support lies inside the interval $(-1,1)$, or on the subspace of functions of the form $\mathbf{1}_{\mathbb{R}^+}\psi$, where $\psi\in\mathcal{S}(\mathbb{R})$ and $\mathbf{1}_{\mathbb{R}^+}$ is the characteristic function of the positive half-line: the former subspace is dense in $L^2(\mathbb{R})$, the last two are not.

The typical course of quantum mechanics for physicists at this stage of the discussion of the mathematical structure introduces the notion of the adjoint $A^\dagger$ of an operator $A$, which can give the opportunity to discuss the density of an observable's domain.

Quite invariably in that context, and in partial abdication to mathematical rigour, $A^\dagger$ is defined through its `matrix elements' as a generalisation of the conjugate transpose of a finite-dimensional matrix \cite[Sect.~8]{Dirac-PrinciplesQM}, \cite[\S 3]{Landau-Lifshitz-3}, \cite[Sect.~1.2-1.3]{sakurai_napolitano_2017}, \cite[Sect.~3.3]{weinberg_2015}, 
%meaning that the action of $A^\dagger$ is only qualified on the elements of the domain of $A$, 
and hermiticity is translated into the property $A=A^\dagger$. This is harmless for bounded and everywhere defined $A$'s, because requiring that $\langle A^\dagger\psi,\phi\rangle=\langle\psi,A\phi\rangle$ for every $\psi,\phi\in\cH$ does characterise, given $A$, the everywhere defined and bounded operator $A^\dagger$: indeed, for a given $\psi\in\cH$ there cannot be two distinct vectors $\xi_1,\xi_2$ both satisfying $\langle\xi_1,\phi\rangle=\langle\psi,A\phi\rangle$ and $\langle\xi_2,\phi\rangle=\langle\psi,A\phi\rangle$ for every $\phi\in\cH$, for otherwise $\langle\xi_1-\xi_2,\phi\rangle=0$ for  every $\phi\in\cH$, whence $\xi_1=\xi_2(=A^\dagger\psi)$. When $A$ is hermitian and unbounded, instead, $\phi$ only runs over a proper subspace $\mathcal{D}(A)\subset\cH$ and in order to characterise unambiguously $A^\dagger\psi$ for some admissible $\psi$ one needs $\mathcal{D}(A)$ to be dense.

This explains why the actual notion of adjoint $A^\dagger$ of $A$ in the general case is
\begin{equation}\label{eq:def-adjoint}
 \begin{split}
  \mathcal{D}(A^\dagger)\;&:=\;\{\psi\in\cH\,|\,\exists\,\xi_\psi\in\cH\textrm{ with }\langle\xi_\psi,\phi\rangle=\langle\psi,A\phi\rangle\;\forall\phi\in\mathcal{D}(A)\} \\
  A^\dagger\psi\;\;&:=\;\xi_\psi\,,
 \end{split}
\end{equation}
and \eqref{eq:def-adjoint} is an unambiguous definition (namely it associates to each such $\psi$ a \emph{unique} $\xi_\psi$) only when $\mathcal{D}(A)$ is \emph{dense} in $\cH$.

\begin{example} Let $\cH=L^2(\mathbb{R})$, $\psi(x)=e^{-x^2}$, and 
 \[
  A\;=\;\textrm{multiplication by $x$}\,,\qquad\mathcal{D}(A)\;=\;
  \left\{
  \phi\in L^2(\mathbb{R})\,\left|
  \begin{array}{c}
   x\phi\in  L^2(\mathbb{R}), \\
   \phi(-x)=\phi(x)
  \end{array}\!\!\right.\right\}\,.
 \]
 Clearly $\psi\in\mathcal{D}(A)$, and $\mathcal{D}(A)$ is not dense in $\cH$ (it is orthogonal to all odd functions). Now, should one want to define $A^\dagger$ on $\psi$ by imposing $\langle A^\dagger\psi,\phi\rangle=\langle\psi,A\phi\rangle$ for every $\phi\in\mathcal{D}(A)$, this would be manifestly ambiguous: for, $\langle\psi,A\phi\rangle=\langle\psi,x\phi\rangle=0$ (scalar product between an even and an odd function), thus $A^\dagger\psi$ could be \emph{any} odd function. 
\end{example}

% \begin{example}~
% 
% \begin{itemize}
%  \item[(i)] Let $\cH=L^2(\mathbb{R})$, $\psi(x)=e^{-x^2}$, and 
%  \[
%   A\;=\;\textrm{multiplication by $x$}\,,\qquad\mathcal{D}(A)\;=\;\{\phi\in L^2(\mathbb{R})\,|\,\phi(-x)=\phi(x)\}\,.
%  \]
%  Clearly $\psi\in\mathcal{D}(A)$, and $\mathcal{D}(A)$ is not dense in $\cH$ (it is orthogonal to all odd functions). Now, should one want to define $A^\dagger$ on $\psi$ by imposing $\langle A^\dagger\psi,\phi\rangle=\langle\psi,A\phi\rangle$ for every $\phi\in\mathcal{D}(A)$, this would be manifestly ambiguous: for, $\langle\psi,A\phi\rangle=\langle\psi,x\phi\rangle=0$ (scalar product between an even and an odd function), thus $A^\dagger\psi$ could be \emph{any} odd function. 
%  \item[(ii)] Definition \eqref{eq:def-adjoint} evidently qualifies $A^\dagger$ also when $A$ is not hermitian (provided that $A$ has dense domain). But in this case $\mathcal{D}(A^\dagger)$ can be as small as a trivial subspace. For instance, for the operator 
%  \[
%    A\psi\:=\:\Big(\int_\mathbb{R}\psi(x)\,\ud x\Big)\psi_0\,,\qquad \mathcal{D}(A)\:=\:\mathcal{S}(\mathbb{R})%\{\textrm{Schwartz functions $\mathbb{R}\to\mathbb{C}$}\}
%  \]
% 
%  
% \end{itemize}
% \end{example}

\subsection{Hermiticity and self-adjointness}\label{sec:hermitian-selfadj}~

All the previous considerations should now make the abstract definition of self-adjoint operator fairly comprehensible and distinguishable from the definition of hermitian operator, even though self-adjointness has not been fully motivated on physical grounds yet (this is going to be the object of Sect.~\ref{sec:emergence_SA}).

\textbf{Hermitian operator}: a linear operator $A$, with domain $\mathcal{D}(A)$, acting in a complex Hilbert space $\cH$ is `hermitian' (with respect to $\cH$) when
\begin{equation}\label{eq:defsa1}
 \langle \psi,A\phi\rangle\;=\;\langle A\psi,\phi\rangle\qquad\forall\psi,\phi\in\mathcal{D}(A)\,.
\end{equation}
Owing to the polarisation identity
\begin{equation*}
\begin{split}
  4\langle \psi,A\phi\rangle\,&=\,\langle \psi+\phi,A(\psi+\phi)\rangle-\langle \psi-\phi,A(\psi-\phi)\rangle \\
  &\quad +\ii\,\langle \psi+\ii\phi,A(\psi+\ii\phi)\rangle-\ii\,\langle \psi-\ii\phi,A(\psi-\ii\phi)\rangle\quad\forall\psi,\phi\in\mathcal{D}(A)\,,
\end{split}
\end{equation*}
condition \eqref{eq:defsa1} is equivalent to
\begin{equation}
 \langle \psi,A\psi\rangle\;=\;\langle A\psi,\psi\rangle\qquad\forall\psi\in\mathcal{D}(A)
\end{equation}
and also equivalent to
\begin{equation}\label{eq:defsa1ter}
 \langle \psi,A\psi\rangle\,\in\,\mathbb{R}\qquad\forall\psi\in\mathcal{D}(A)\,.
\end{equation}
If in addition $\mathcal{D}(A)$ is dense, and therefore one can give meaning to the adjoint $A^\dagger$, one sees from definition \eqref{eq:def-adjoint} that the hermiticity of $A$ is tantamount as
\begin{equation}
 \mathcal{D}(A)\,\subset\,\mathcal{D}(A^\dagger)\qquad\textrm{and}\qquad A\psi=A^\dagger\psi\quad\forall\psi\in\mathcal{D}(A)\,.
\end{equation}

\textbf{Self-adjoint operator}: a linear operator $A$, with dense domain $\mathcal{D}(A)$, acting in a complex Hilbert space $\cH$ is `self-adjoint' (with respect to $\cH$) if $A=A^\dagger$, meaning that
\begin{equation}
  \mathcal{D}(A)\,=\,\mathcal{D}(A^\dagger)\qquad\textrm{and}\qquad A\psi=A^\dagger\psi\quad\forall\psi\in\mathcal{D}(A)\,=\,\mathcal{D}(A^\dagger)\,.
\end{equation}
A self-adjoint operator is therefore (densely defined and) hermitian, whereas the opposite in general is not true. Both a densely defined hermitian operator $A$ and a self-adjoint operator $A$ satisfy the fact that they have the same formal action on the vectors of $\mathcal{D}(A)$, but only when $A$ is self-adjoint do the two domains $\mathcal{D}(A)$ and $\mathcal{D}(A^\dagger)$ coincide.

\begin{example}\label{ex:momentum-in-ab}
 Let $\cH=L^2(a,b)$ with $-\infty<a<b<+\infty$ and 
\begin{align*}
A_1&=-\ii\frac{\ud}{\ud x}           &  \mathcal{D}(A_1) &=\{\psi\in C^1[a,b]\,|\,\psi(a)=\psi(b)=0\}     \\
A_2&=-\ii\frac{\ud}{\ud x}           &  \mathcal{D}(A_2) &=\{\psi\in \mathscr{H}^1(a,b)\,|\,\psi(a)=\psi(b)\}\\
A_3&=-\ii\frac{\ud}{\ud x}           &  \mathcal{D}(A_3) &=\mathscr{H}^1(a,b)\,.
\end{align*}
(Let us recall that $\mathscr{H}^1(a,b)$ is the subspace of $L^2(a,b)$ of functions that are (weakly) differentiable on $(a,b)$ and such that their derivative is still square-integrable; as a matter of fact each element of $\mathscr{H}^1(a,b)$ turns out to be an absolutely continuous function on $[a,b]$ and hence it makes sense to evaluate it at $x=a$ and $x=b$.) One has $\mathcal{D}(A_1)\varsubsetneq\mathcal{D}(A_2)\varsubsetneq\mathcal{D}(A_3)$, and all three subspaces are dense in $L^2(a,b)$. The formal action of the operator is the same in all cases, but the three operators are profoundly different: for example $A_1$ has no eigenvector at all (there is no non-zero solution to $A_1\psi=\lambda\psi$ with $\psi\in\mathcal{D}(A_1)$), whereas $A_2$ has an orthonormal basis of eigenvectors. Working out the definition \eqref{eq:def-adjoint} in this case (see, e.g., \cite[Sect.~2.2]{Teschl-MMQM-2014} for details) one finds: $A_1^\dagger=A_3$, therefore $A_1$ is hermitian but not self-adjoint; $A_2^\dagger=A_2$, therefore $A_2$ is self-adjoint; $A_3$ is not even hermitian, for 
\begin{align*}
 A_3^\dagger&=-\ii\frac{\ud}{\ud x}           &  \mathcal{D}(A_3) &=\{\psi\in \mathscr{H}^1(a,b)\,|\,\psi(a)=\psi(b)=0\}\,.
\end{align*}
\end{example}

\subsection{Closed operators. Self-adjoint operators are closed}\label{sec:closedoperators}~

There is in fact another technical feature that somehow naturally pops up when dealing, among others, with unbounded hermitian operators on Hilbert space: the possible \emph{closedness} of an operator. This is a notion that is reasonable to conceive  after realising that certain operators are only densely defined (not everywhere defined) and unbounded: this is why as somewhat technical as it may appear, it is natural to flash it at this stage of our discussion. Its relevance in the physically grounded path for self-adjointness will emerge in Sect.~\ref{sec:emergence_SA}.

If $A$ is (everywhere defined and) bounded on $\cH$, even not hermitian, an elementary argument based on the linearity of $A$ shows that $A$ is continuous, and in fact boundedness and continuity are equivalent: thus, for a sequence $(\psi_n)_{n\in\mathbb{N}}$ and a vector $\psi$, all in $\cH$, if $\psi_n\to\psi$ as $n\to\infty$, then also $A\psi_n\to A\psi$ (the limits being in the Hilbert norm).

For unbounded operators continuity is lost, but a similar, albeit weaker, property that \emph{sometimes} is still satisfied is the following: if, for a sequence $(\psi_n)_{n\in\mathbb{N}}$ in $\mathcal{D}(A)$ and two vectors $\psi,\phi\in\cH$ one has $\psi_n\to\psi$ and $A\psi_n\to\phi$, then $\psi\in\mathcal{D}(A)$ and $A\psi=\phi$. At first sight this seems such an obvious requirement to be fulfilled, in particular by operators representing quantum observables, but in general it is not. When it is, $A$ is said to be `closed'. An everywhere defined bounded operator is obviously also closed. 
%It is also straightforward to see from the above definition that the closedness of the operator $A$ is equivalent to the completeness of the vector space $\mathcal{D}(A)$ with respect to the new norm $\|\psi\|_A:=(\|\psi\|^2+\|A\psi\|^2)^{1/2}$.

\begin{example}\label{ex:closed-nonclosed}
 With respect to $\cH=L^2(\mathbb{R})$ consider the ``position operators''
 \begin{align*}
A_1&=\textrm{multiplication by $x$},          &  \mathcal{D}(A_1) &=\{\psi\in L^2(\mathbb{R})\,|\,\mathrm{supp}(\psi)\textrm{ is a compact subset of }\mathbb{R}\},  \\
A_2&=\textrm{multiplication by $x$},            &  \mathcal{D}(A_2) &=\{\psi\in L^2(\mathbb{R})\,|\,x\psi\in L^2(\mathbb{R})\}\,.
\end{align*}
 %As customary, $C^\infty_0(\mathbb{R})$ denotes the space of functions that are differentiable (in the classical sense) infinitely many times and such that their support is a compact subset of $\mathbb{R}$. 
 Here $\mathrm{supp}(\psi)$ is the support of the function $\psi$. Clearly, $\mathcal{D}(A_1)\subset\mathcal{D}(A_2)$ and it is a standard  fact from functional analysis that both such domains are dense in $L^2(\mathbb{R})$, although we will not need that. Both $A_1$ and $A_2$ are manifestly hermitian. Now, set
 \[
  \begin{split}
   \psi_n(x)\;&:=\;\mathbf{1}_{[-n,n]}(x)\,e^{-x^2}\qquad (n\in\mathbb{N}) \\
    \psi(x)\;&:=\;e^{-x^2}\,,
  \end{split}
 \]
 where $\mathbf{1}_{[-n,n]}$ is the characteristic function of the interval $[-n,n]$. One has $\psi_n\in\mathcal{D}(A_1)$ and moreover, by dominated convergence, $\psi_n\xrightarrow[]{(L^2)}\psi$ and $x\psi_n\xrightarrow[]{(L^2)}x\psi$ as $n\to\infty$. However, $\psi\notin\mathcal{D}(A_1)$: the operator $A_1$ is \emph{not} closed. Conversely, the operator $A_2$ \emph{is} closed. To see that, assume that for a sequence $(\psi_n)_{n\in\mathbb{N}}$ in $\mathcal{D}(A_2)$ and two functions $\psi,\phi\in L^2(\mathbb{R})$ one has $\psi_n\xrightarrow[]{(L^2)}\psi$ and $x\psi_n\xrightarrow[]{(L^2)}\phi$.  As $L^2$-limits are also point-wise (almost everywhere) limits, then $\phi(x)=\lim_{n\to\infty}x\psi_n(x)=x\psi(x)$ for almost every $x$, whence $x\psi=\phi\in L^2(\mathbb{R})$. This means precisely that $\psi\in\mathcal{D}(A_2)$ and $A_2\psi=\phi$, therefore $A_2$ is closed.
 \end{example}

For the pedagogical path that we are elaborating here there is no need to discuss the basics of the general theory of closed operators on Hilbert space (see, e.g., \cite[Chapters 1--3]{schmu_unbdd_sa}, but for one property that it is important to highlight.

\begin{lemma}\label{lem:Aselfadj-is-closed}
 If $A$ is self-adjoint on $\cH$, then $A$ is closed.
\end{lemma}

The proof is rather simple.

\begin{proof}[Proof of Lemma \ref{lem:Aselfadj-is-closed}]
 By assumption $A=A^\dagger$. Let us show that $A^\dagger$ is closed. Assume that for a sequence $(\psi_n)_{n\in\mathbb{N}}$ in $\mathcal{D}(A^\dagger)$ and two vectors $\psi,\phi\in\cH$ one has $\psi_n\to\psi$ and $A^\dagger\psi_n\to\phi$. Then, for any $\xi\in\mathcal{D}(A)$ one has
 \[
  \langle\phi,\xi\rangle\;=\;\lim_{n\to\infty}\langle A^\dagger\psi_n,\xi\rangle\;=\;\lim_{n\to\infty}\langle\psi_n,A\xi\rangle\;=\;\langle\psi,A\xi\rangle\,.
 \]
 In view of definition \eqref{eq:def-adjoint}, $\psi\in\mathcal{D}(A^\dagger)$ and $A^\dagger\psi=\phi$. Thus, $A^\dagger$ is closed.
\end{proof}

Summarising: generic hermitian operators are not necessarily closed; self-adjoint operators always are.

%The closedness of $A$ is therefore the closedness of the so-called `graph' of $A$, namely the set $\{(\psi,A\psi)\,|\,\psi\in\mathcal{D}(A)\}$ as a subset of 

\subsection{Algebraic manipulation of unbounded quantum observables is subject to domain issues}\label{sec:domains-touchy}~

At the end of the current Section, and prior to embarking on the core part of our pedagogical path (Section \ref{sec:emergence_SA}), we find instructive to highlight a circumstance inescapably connected with the conclusion that quantum observables, as hermitian (and, eventually, actually more: self-adjoint) operators on Hilbert space, are defined on a domain of hermiticity (eventually: of self-adjointness) which is only a proper dense subspace of the underlying Hilbert space whenever the observable is unbounded, and which is autonomously declared together with the formal action of the operator.

We refer here to the circumstance that deceptively innocent algebraic manipulations of unbounded quantum observables, such as those manipulations that are most common in physics (think of the formulation of the canonical commutation relation ``$QP-PQ=\ii\hbar$'') are indeed a touchy business.

In the bounded case, the family $\mathcal{B}(\cH)$ of (everywhere defined and) bounded operators on a given Hilbert space $\cH$ is an algebra with respect to the natural operator sum and operator multiplication (in fact, $\mathcal{B}(\cH)$ displays interlaced algebraic and analytic properties that give rise to a much richer structure, a $C^*$-algebra \cite[Sect.~2.1]{Bratteli-Robinson-1}). In the unbounded case, sums and products of hermitian (self-adjoint) operators are still of course under control: only, one must take into account domain issues of all sort, overlooking which, one easily falls into paradoxes (the best scenario, as at least paradoxes are manifest) or erroneous conclusions.

\begin{example}\label{ex:paradox1} Consider the one-dimensional position and momentum observables, acting formally as $Q\psi=x\psi$ and $P\psi=-\ii\psi'$, on the Hilbert space $\cH=L^2(\mathbb{R})$, and the functions $\psi_p$, $p\in\mathbb{R}$, on the real line defined by $\psi_p(x):=e^{\ii p x}$. Since $P\psi_p=p\psi_p$, a formal computation yields
 \[
  \begin{split}
   \langle\psi_p,[Q,P]\psi_p\rangle\;&=\;\langle\psi_p,QP\psi_p\rangle-\langle\psi_p,PQ\psi_p\rangle\;=\;p\langle\psi_p,Q\psi_p\rangle-\langle P\psi_p,Q\psi_p\rangle \\
   &=\; p\langle\psi_p,Q\psi_p\rangle-p\langle\psi_p,Q\psi_p\rangle\;=\;0\,.
  \end{split}
 \]
 This does \emph{not} disprove the standard canonical commutation relation between $P$ and $Q$. Indeed, the $\psi_p$'s do not belong to $\cH$ and hence the above expressions are not scalar products in $\cH$. Moreover, from the actual domains of self-adjointness of $Q$ and $P$ in $\cH$, that can be proved to be
 \[
  \begin{split}
   \mathcal{D}(Q)\;&=\;\{\psi\in L^2(\mathbb{R}\,|\,x\psi\in L^2(\mathbb{R}\}\,, \\
   \mathcal{D}(P)\;&=\; \mathscr{H}^1(0,1)\;=\;\{\psi\in L^2(\mathbb{R}\,|\,\psi'\in L^2(\mathbb{R}\}\,,
  \end{split}
 \]
 it is clear that the observable $P$ admits no eigenfunction in $\cH$.
\end{example}

\begin{example}
 Consider a free quantum particle in the one-dimensional box $[0,1]$, thus with Hilbert space $\cH=L^2(0,1)$, and Hamiltonian given by the self-adjoint free energy Dirichlet operator $H_D$ already introduced in Example \ref{ex:bc}. Defining $\psi(x):=x(1-x)$, then clearly $\psi\in\mathcal{D}(H_D)$ and $(H_D\psi)(x)=2$. Thus,
 \[
  \langle H_D\psi,H_D\psi\rangle\;=\;4\,.
 \]
 On the other hand, though, the fourth derivative of $\psi$ is the zero function: interpreting this as $H_D^2\psi\equiv 0$, and exploiting the self-adjointness of $H_D$, one would be led to conclude
 \[
   \langle H_D\psi,H_D\psi\rangle\;=\;\langle\psi,H_D^2\psi\rangle\;=\;0\,.
 \]
 The apparent contradiction is due to the fact that the function $H_D\psi$, namely the constant function identically equal to 2, does \emph{not} belong to $\mathcal{D}(H_D)$ (it fails to satisfy Dirichlet boundary conditions), hence one cannot further apply $H_D$ to $H_D\psi$ (or, in other words, $\psi\notin\mathcal{D}(H_D^2)$). In the second computation the expression $\langle\psi,H_D^2\psi\rangle$ is therefore meaningless.
 \end{example}

 \begin{example}\label{ex:momentum-box-2}
  Consider the momentum observable for a quantum particle in the one-dimensional box $[0,1]$, defined as the self-adjoint operator $P_\theta$ already introduced in Example \ref{ex:momentum-box} for some fixed $\theta\in[0,2\pi)$, that is,
  \[
   P_\theta\;=\;-\ii\frac{\ud}{\ud x}\,,\qquad\mathcal{D}(P_\theta)\;=\;\{\psi\in \mathscr{H}^1(0,1)\,|\,\psi(1)=e^{\ii\theta}\psi(0)\}\,.
  \]
  (The hermiticity of $P_\theta$ follows by integration by parts; the actual self-adjointness $P_\theta=P_\theta^\dagger$ may be proved by working out the definition \eqref{eq:def-adjoint} for $P_\theta^\dagger$ or by applying the standard criterion of self-adjointness \cite[Theorem VIII.3]{rs1}, as we tacitly do in some of the mathematical proofs of Sect.~\ref{sec:math-proofs} below.) Consider also the orthonormal basis $(\psi_n)_{n\in\mathbb{N}_0}$ of $\cH$ given by $\psi_n(x):=\sqrt{2}\cos\pi n x$. The matrix elements of $P_\theta$ with respect to such basis are
  \[
   p_{n,m}\;:=\;\langle\psi_n,P_\theta\psi_m\rangle\;=\;-\ii\int_0^1\psi_n(x)\,\psi_m'(x)\,\ud x\,,\qquad n,m\in\mathbb{N}_0\,.
  \]
  Now, integration by parts yields
  \[
   \overline{p_{n,m}}\;=\;p_{m.n}+\ii\big(\psi_m(1)\psi_n(1)-\psi_m(0)\psi_n(0)\big)\,,
  \]
  therefore $\overline{p_{n,m}}\neq p_{m.n}$ whenever $n+m$ is an odd integer, as if the matrix representing $P_\theta$ was not hermitian. This does \emph{not} disprove the hermiticity, though. Indeed, as follows from the relation $\psi_n(1)=(-1)^n\psi_n(0)$, the basis $(\psi_n)_{n\in\mathbb{N}_0}$ cannot be entirely contained in $\mathcal{D}(P_\theta)$, whatever the initial choice of $\theta$, and for those $\psi_m$'s not belonging to $\mathcal{D}(P_\theta)$ the above expression $\langle\psi_n,P_\theta\psi_m\rangle$ is meaningless.  
 \end{example}

\section{Second part: emergence of self-adjointness for quantum observables}\label{sec:emergence_SA}

Let us enter the central part of our pedagogical path: the discussion on physically grounded motivations that qualify quantum observables mathematically as self-adjoint, and not merely hermitian operators.

We should like to start with the most relevant class of quantum observables, the Hamiltonians (the generators of the quantum dynamics), and then proceed with other classes of observables for which certain physical features translate mathematically into self-adjointness.

\subsection{Self-adjointness of the quantum Hamiltonian inferred from the Schr\"{o}dinger equation}\label{sec:selfadj-spectralthm}~

The most relevant type of quantum observable for which to discuss the emergence and the role of self-adjointness is surely the Hamiltonian of a given quantum system.

In the typical physical introduction to quantum mechanics the Hamiltonian emerges as a distinguished operator in connection with the time evolution of the system.

More precisely \cite[Sect.~27]{Dirac-PrinciplesQM}, \cite[Sect.~2.1]{sakurai_napolitano_2017}, \cite[Sect.~3.6]{weinberg_2015}, by means of subtle physical reasonings one argues that a quantum system evolves in time along a trajectory $\psi(t)$ of states of the considered Hilbert space $\cH$ of the form $\psi(t)=U(t)\psi(0)$, where $\{U(t)\,|\,t\in\mathbb{R}\}$ is a collection of everywhere defined bounded operators constituting, technically speaking, a `strongly continuous one-parameter unitary group'. That is:
\begin{itemize}
 \item[1.] each $U(t)$ preserves the norm $\|U(t)\psi\|=\|\psi\|$ of every $\psi\in\cH$ and its range is the whole $\cH$, equivalently, $U(t)^\dagger U(t)=\mathbbm{1}=U(t)U(t)^\dagger$ (\emph{unitarity});
 \item[2.] the composition rule $U(t_1)U(t_2)=U(t_1+t_2)$ is satisfied for every instant $t_1,t_2\in\mathbb{R}$ with $U(0)=\mathbbm{1}$, the identity operator (\emph{group composition});
 \item[3.] $\|U(t)\psi-\psi\|\to 0$ as $t\to 0$ for every $\psi\in\cH$ (\emph{strong continuity}).
\end{itemize}
In turn, to the collection of the $U(t)$'s one associates an operator $H$ obtained from the $O(t)$-term of a formal analytic expansion of $U(t)$ as $t\to 0$ and finally argues that $\psi(t)$ is determined by the celebrated Schr\"{o}dinger equation
\begin{equation}\label{eq:schreq}
 \ii\frac{\ud}{\ud t}\psi(t)\;=\;H\psi(t)
\end{equation}
that governs the evolution of the quantum system. Besides, first-quantisation-like arguments prescribe the explicit form of $H$ from its classical counterpart in those cases when the system is described by a wave-function $\psi(t;x)$ of space-time coordinates only: for concreteness, when $\cH=L^2(\mathbb{R})$ 
%(or $\cH=L^2[a,b]$ and the like) 
the Hamiltonian of a one-particle quantum system has the form $H=-\frac{\ud^2}{\ud x^2}+V(x)$, a Schr\"{o}dinger operator with real-valued potential $V$. The formal hermiticity of such $H$ follows as usual from integration by parts tested on a suitable class of functions that be regular and fast decaying to a degree that depends on the potential $V$.

A crucial fact for the mathematical framework of quantum mechanics, as we shall now discuss, is that the Schr\"{o}dinger dynamics does have the above-mentioned properties of unitarity, continuity in time, and group composition at subsequent times \emph{if and only if} the Hamiltonian $H$ governing the Schr\"{o}dinger equation \eqref{eq:schreq} is self-adjoint (and not merely hermitian).

For an efficient way to present this point of view we find the following formulation as the most convenient -- let us postpone the proof to Sect.~\ref{sec:math-proofs}.

\begin{theorem}\label{thm:selfadj-SchrEq}
 Let $\cH$ be a complex Hilbert space and let $H$ be a \emph{hermitian} operator acting on $\cH$ with domain $\mathcal{D}(H)\subset\cH$. The two conditions (i) and (ii) below are equivalent.
 \begin{itemize}
  \item[(i)] There exists a strongly continuous one-parameter unitary group  $\{U(t)\,|\,t\in\mathbb{R}\}$ acting on $\cH$ such that
 \begin{itemize}
  \item[$\bullet$] for every $t\in\mathbb{R}$ one has $U(t)\mathcal{D}(H)\subset\mathcal{D}(H)$,
  \item[$\bullet$] for every $\psi_0\in\mathcal{D}(H)$ the collection of vectors $\psi(t):=U(t)\psi_0$ defined for every $t\in\mathbb{R}$ constitute a solution to the problem 
  \begin{equation}\label{eq:SchrIVP}
   \begin{cases}
    \ii\frac{\ud}{\ud t}\psi(t)\;=\;H\psi(t) \\
    \psi(0)\;=\;\psi_0
   \end{cases}
  \end{equation}
 \end{itemize}
  \item[(ii)] The operator $H$ is self-adjoint.
 \end{itemize}
  When either condition above is satisfied, one has the following:
  \begin{itemize}
   \item[1.] the domain and the action of $H$ satisfy
  \begin{equation}\label{eq:actdomH}
   \begin{split}
    \mathcal{D}(H)\;&=\;\left\{\psi\in\cH\,\left|\,\exists\,\frac{\ud}{\ud t}\Big|_{t=0}U(t)\psi:=\lim_{t\to 0}\frac{U(t)-\mathbbm{1}}{t}\psi\in\cH \right.\right\} \\
    H\psi\;&=\;\ii\frac{\ud}{\ud t}\Big|_{t=0}U(t)\psi
   \end{split}
  \end{equation}
  and in particular $\mathcal{D}(H)$ is dense in $\cH$;
  \item[2.] one has
  \begin{equation}\label{eq:HU=UH=dU}
   HU(t)\psi_0\;=\;U(t)H\psi_0\;=\;\ii\frac{\ud}{\ud t}U(t)\psi_0\qquad\forall \psi_0\in\mathcal{D}(H)\,,\;\forall t\in\mathbb{R}\,;
  \end{equation}
  \item[3.] there exists a unique solution $\psi(\cdot)\in C^1(\mathbb{R},\cH)$, with $\psi(t)\in\mathcal{D}(H)$ $\forall t\in\mathbb{R}$, to the problem \eqref{eq:SchrIVP}, and it is given precisely by $\psi(t)=U(t)\psi_0$.
  \end{itemize}
\end{theorem}

Observe that part (i) of the statement of Theorem \ref{thm:selfadj-SchrEq} collects precisely the \emph{physical requirements} on the quantum dynamics.

The above implication (i) $\Rightarrow$ (ii) is essentially the celebrated theorem by Stone (see Theorem \ref{thm:Stone} below) applied to a strongly continuous unitary group that is linked with the operator $H$ according to the assumptions stated in (i). For this, it is crucial to assume $H$ to be hermitian in the first place, which, as said, has already its own physical motivation (reality of the expectations). One does not need to assume $H$ to be densely defined, though: this follows from the requirements (i).

If $H$ is hermitian but not self-adjoint, Theorem \ref{thm:selfadj-SchrEq} implies that solutions $\psi(\cdot)$ to the initial value problem \eqref{eq:SchrIVP} fail to simultaneously satisfy all the conditions in (i). In Subsect.~\ref{sec:non-uniquedynamics} below we shall further discuss such an occurrence.

Summarising, self-adjointness is that feature of quantum Hamiltonians (mere hermiticity would not suffice) ensuring the well-posedness of the Schr\"{o}dinger equation's initial value problem \eqref{eq:SchrIVP}, namely the existence of a unique solution $\psi(\cdot)$ in $C^1(\mathbb{R},\cH)$ with values in $\mathcal{D}(H)$ and strongly continuous in the initial datum $\psi_0$, which moves in time along the unitary dynamics.

In retrospect, as the proof of Theorem \ref{thm:selfadj-SchrEq} shows (see Sect.~\ref{sec:math-proofs}), such a solution has the form $\psi(t)=e^{-\ii t H}\psi_0$, and the self-adjointness of $H$ is precisely the property that allows one to give meaning to the unitary operator $e^{-\ii t H}$. For finite-dimensional matrices or bounded operators on an infinite-dimensional Hilbert space one may construct their exponential as the operator-norm-convergent series
\begin{equation}\label{eq:expitH}
 e^{-\ii t H}\;=\;\sum_{n=0}^\infty\frac{(-\ii)^n t^n}{n!}H^n\qquad\textrm{($H$ bounded)}\,,
\end{equation}
but when $H$ is unbounded the above series cannot make sense on the whole $\cH$ (there are surely vectors $\psi$ of $\cH$ not belonging to $\mathcal{D}(H^n)$ for some $n$) and hence does not define a unitary operator on $\cH$. The possibility of realising non-ambiguously and consistently operators like  $e^{-\ii t H}$, or more generally like $f(H)$ for a suitable class of functions on $\mathbb{R}$, is guaranteed by that mathematical apparatus that goes under the collective name of `spectral theorem and functional calculus for self-adjoint operators' (see, e.g., \cite[Chapters 4 and 5]{schmu_unbdd_sa}). Thus, in the present context quantum Hamiltonians need be self-adjoint because self-adjointness (and not mere hermiticity) guarantees a convenient functional calculus so that the Hamiltonian generates a meaningful Schr\"{o}dinger evolution.

\subsection{Self-adjointness for building the unitary $e^{-\ii t H}$ by series expansion on a dense subspace of vectors}\label{sec:analyticvectors}~

This Subsection is a detour from the main line of our reasoning and serves as a complement to the arguments of the previous Subsection.

Indeed, as in physical contexts the temptation is strong to still give meaning to $e^{-\ii t H}$ through a formal series of the type \eqref{eq:expitH} even when the Hamiltonian $H$ is an unbounded hermitian operator, in the framework of our pedagogical path for self-adjointness a few more observations would be instructive.

We have already commented, in view of the non-triviality of the domain $\mathcal{D}(H)$ when $H$ is hermitian and unbounded, that not on all vectors of $\cH$ can one apply higher powers of $H$. The next most reasonable attempt is to give meaning to the series
\begin{equation}\label{eq:exponpsi}
  \sum_{n=0}^\infty\frac{(-\ii)^n t^n}{n!}H^n\psi
\end{equation}
as a convergent series (in the Hilbert space norm) at least for a convenient \emph{selection} of vectors $\psi$ for which one can prove that \eqref{eq:exponpsi} defines an action $\psi\mapsto\textrm{``}e^{-\ii t H}\textrm{''}\psi$ that preserves the vector norm, has inverse $\psi\mapsto\textrm{``}e^{\ii t H}\textrm{''}\psi$, and has the group composition properties in $t$.

Such $\psi$'s have of course to constitute a linear subspace and to belong to $\mathcal{D}(H)$; moreover, their linear span need be \emph{dense} in $\cH$: only in this case the ``temporary'' operator $\textrm{``}e^{-\ii t H}\textrm{''}$ defined by series on a dense subspace of $\cH$ can be consistently extended to an everywhere defined bounded (and unitary, because of the above properties) operator on $\cH$ -- a construction that is canonical, relies crucially on the completeness of $\cH$ as a Hilbert space, and is customarily referred to as `continuous linear extension' or `B.L.T.~theorem' (see, e.g., \cite[Theorem I.7]{rs1}).

Now, it turns out that requiring $\mathcal{D}(H)$ to contain a subspace of distinguished vectors with the properties listed above essentially qualifies $H$ as a self-adjoint operator, which makes self-adjointness inescapable also along this line of reasoning.

The precise formulation of such fact is made in terms of so-called `analytic vectors'. Given an operator $A$ on Hilbert space $\cH$, an element $\psi\in\cH$ is called an analytic vector for $A$ when  $\psi\in\mathcal{D}(A^n)$ for every $n\in\mathbb{N}$ and 
\[
 \sum_{n=0}^\infty\frac{\|A^n\psi\|}{n!}t^n\;<\;+\infty
\]
for some $t>0$. This requisite is clearly designed to be a sufficient condition for the convergence of the series \eqref{eq:exponpsi}.

This leads finally to the announced characterisation of self-adjointness as that feature that consists of having ``sufficiently many'' analytic vectors within the operator domain (for the proof of which we refer, e.g., to \cite[Theorem X.39]{rs2} or \cite[Theorem 7.16]{schmu_unbdd_sa}).

\begin{theorem}[Nelson's analytic vector theorem]\label{thm:nelson}
 Let $A$ be a \emph{hermitian} and \emph{closed} operator on a Hilbert space $\cH$. Then these two conditions are equivalent:
 \begin{itemize}
  \item[(i)] $\mathcal{D}(A)$ contains a dense set of analytic vectors for $A$;
  \item[(ii)] $A$ is self-adjoint.
 \end{itemize}
\end{theorem}

Informally speaking, the catch then is: when one requires the hermitian quantum Hamiltonian $H$ to allow for sufficiently many vectors $\psi$ in its domain so as to guarantee that $\sum_{n=0}^\infty\frac{\|H^n\psi\|}{n!}t^n<+\infty$ for some $t>0$ and hence to define \eqref{eq:exponpsi} as a Hilbert-norm-convergent series, where ``sufficiently many'' means more precisely a dense of them so as to finally define $e^{-\ii t H}$ on the whole $\cH$ by continuous linear extension from \eqref{eq:exponpsi}, in practice one is requiring exactly that $H$ be self-adjoint. (We wrote ``in practice'' because of the technical caveat that $H$ be closed, in order to apply Theorem \ref{thm:nelson}.)

For a hermitian and non-self-adjoint $H$, such a construction is not possible. A generic hermitian operator $H$ on Hilbert space $\cH$ may have no analytic vectors at all (apart obviously the zero vector). Even more: it may have no non-zero vector $\psi$ simultaneously belonging to $\mathcal{D}(H^n)$ for every $n\in\mathbb{N}$.

\begin{example}
 In contrast to the one-dimensional self-adjoint position observable $Q$, namely
 \[
  \mathcal{D}(Q)\,=\,\{\psi\in L^2(\mathbb{R})\,|\,x\psi\in L^2(\mathbb{R})\}\,,\qquad Q\psi\,=\,x\psi
 \]
 (see Examples \ref{ex:position-op-schwartz-simple}, \ref{ex:closed-nonclosed}, and \ref{ex:paradox1} above), consider the operators
 \begin{align*}
   \mathcal{D}(Q_1)\,&=\,\{\textrm{really simple functions $\mathbb{R}\to\mathbb{C}$}\}\,, & Q_1\psi\,=\,x\psi\,, \\
   \mathcal{D}(Q_2)\,&=\,C^\infty_0(\mathbb{R})\,, & Q_2\psi\,=\,x\psi\,.
 \end{align*}
 Both $Q_1$ and $Q_2$ have dense domain and are hermitian (Examples \ref{ex-position-cInf0-S} and \ref{ex:position-op-schwartz-simple}). Yet, one deduces from definition \eqref{eq:def-adjoint} that $Q_1^\dagger=Q_2^\dagger=Q$, therefore \emph{neither} $Q_1$ \emph{nor} $Q_2$ is self-adjoint. (In a sense, both $Q_1$ and $Q_2$ are really close to be self-adjoint: technically speaking, they are `\emph{essentially self-adjoint}', but we will not need this information.) Now, any really simple function $\psi$, once is multiplied by $x$, yields a function $x\psi$ that is not a really simple function any more (multiplication by $x$ obviously destroys the step-function structure). Therefore,
 \[
  \mathcal{D}(Q_1^2)\;=\;\{\psi\in\mathcal{D}(Q_1)\,|\,Q_1\psi\in\mathcal{D}(Q_1)\}\;=\;\{0\}\,.
 \]
 As a consequence, there are no (non-zero) analytic vectors for $Q_1$. Theorem \ref{thm:nelson} is not applicable to $Q_1$ -- and indeed, as observed above, $Q_1$ fails to be self-adjoint. On the other hand, any $\psi\in\mathcal{D}(Q_2)$ is an analytic vector for $Q_2$, as
 \[
  \sum_{n=0}^\infty\frac{\;\|x^n\psi\|_{L^2(\mathbb{R})}}{n!}\,\leqslant\,\|\psi\|_{L^\infty(\mathbb{R})}\,|\mathrm{supp}(\psi)|^{\frac{1}{2}}\sum_{n=0}^\infty\frac{1}{n!}\Big(\sup_{x\in\mathrm{supp}(\psi)}|x|\Big)^n\,<\,+\infty
 \]
 (here $|\mathrm{supp}(\psi)|$ is the finite measure of the support of $\psi$); however, $Q_2$ is not closed (Example \ref{ex:closed-nonclosed}), hence Theorem \ref{thm:nelson} is not applicable to $Q_2$ either.
 Instead, $Q$ is closed (Example \ref{ex:closed-nonclosed}) \emph{and}, as just argued, it admits plenty of analytic vectors in its domain, actually the whole dense subspace $C^\infty_0(\mathbb{R})$. 
\end{example}

\subsection{Non-uniqueness of Schr\"{o}dinger's dynamics when self-adjointness is not declared}\label{sec:non-uniquedynamics}~

We find instructive at this stage to further complement the analysis of Subsect.~\ref{sec:selfadj-spectralthm} by discussing the \emph{physically unacceptable} occurrence of \emph{non-unique} Schr\"{o}\-ding\-er dynamics in the lack of a definite declaration of self-adjointness for the Hamiltonian.

To develop this point, let us shift our focus onto the time-dependent \emph{differential equation} by which the Schr\"{o}dinger dynamics is actually formulated, therefore only considering the formal action of the Hamiltonian acting therein -- typically a differential operator on the appropriate $L^2$-space. Misleadingly enough, a very common physical claim is that the Schr\"{o}dinger equation ``encodes'' all the information to determine \emph{the} forward-in-time trajectory $\{\psi(t)\,|\,t\geqslant 0\}$ starting from a given initial state of the system. Whereas this statement is ``morally'' true, it cannot be valid based on the sole differential equation: if an explicit declaration of self-adjointness of the Hamiltonian is lacking, and in particular if one does not make the additional prescription that the solution $\psi(t)$ must evolve unitarily inside a given domain of self-adjointness, then the differential equation alone may well give rise to an infinite multiplicity of solutions, all with the same initial state.

In fact, a preliminary level where non-uniqueness emerges is related to the obvious (and fundamental) physical request that solutions to the Schr\"{o}dinger equation belong at any later time $t>0$ to the Hilbert space under consideration.

\begin{example}
 Let us take for concreteness $\cH=L^2(\mathbb{R})$ (or more generally $\cH=L^2(\mathbb{R}^d)$ with $d\in\mathbb{N}$) and let us regard the Schr\"{o}dinger equation
\begin{equation}\label{eq:SchrEqV}
 \ii\frac{\partial}{\partial t}\psi(t,x)\,=\,-\frac{\partial^2}{\partial x^2}\psi(t,x)+V(x)\psi(t,x)
\end{equation}
as a partial differential equation. Let the potential $V:\mathbb{R}\to\mathbb{R}$ be, for simplicity, an analytic function. 
Upon re-interpreting \eqref{eq:SchrEqV} as $P(-\ii\frac{\partial}{\partial t},-\ii\frac{\partial}{\partial x},x)\psi=0$ with the second order differential operator defined by $P(\xi_1,\xi_2,x):=\xi_1+\xi_2^2+V(x)$, one sees that in the standard sense of linear partial differential operator theory the initial value problem for $P(-\ii\frac{\partial}{\partial t},-\ii\frac{\partial}{\partial x},x)\psi=0$ with initial condition $\psi(0,x)=\psi_0(x)$ for some nice function $\psi_0\in\mathcal{S}(\mathbb{R})\subset L^2(\mathbb{R})$ is a `characteristic initial value problem' in the sense that the initial plane $t=0$ is `characteristic' for the equation (see, e.g., \cite[Chapter III]{Hormander-LinearPDOp-1976}). As such, the initial value problem has an \emph{infinite number} of solutions. A clever example of an infinity of non-zero solutions to \eqref{eq:SchrEqV} with $V\equiv 0$, namely to the free Schr\"{o}dinger equation, with zero initial value at $t=0$, is discussed in \cite[Theorem 8.9.2]{Hormander-LinearPDOp-1976}: however, \emph{none} of such (non-zero) solutions belongs to $L^2(\mathbb{R})$ for any $t>0$! 
In other words: not imposing solutions to be square-integrable allows for an unphysical multiplicity of solutions for the Schr\"{o}dinger dynamics.
\end{example}

Of course a physically meaningful solution must describe a trajectory inside the considered Hilbert space (such as $L^2(\mathbb{R})$ for solutions to \eqref{eq:SchrEqV}), and we therefore assume that such a request shall be tacitly made henceforth. Yet, this does not fix the issue of non-uniqueness, if an explicit declaration of self-adjointness is not made.

To discuss this point, consider the Hilbert space $\cH=L^2(0,1)$ (in the end we shall comment on the analogous situation for $\cH=L^2(\mathbb{R}^d)$), and the free kinetic energy Hamiltonian $-\frac{\ud^2}{\ud x^2}$ in one of the following distinct realisations as operators on $\cH$:
\begin{align*}
H_\circ&=-\frac{\ud^2}{\ud x^2}           &  \mathcal{D}(H_\circ) &=C^\infty_0(0,1)\,,     \\
H_D&=-\frac{\ud^2}{\ud x^2}           &  \mathcal{D}(H_D) &=\big\{ \psi\in \mathscr{H}^2(0,1)\,\big|\, \psi(0)=0=\psi(1) \big\}\,,     \\
H_P&=-\frac{\ud^2}{\ud x^2}           &  \mathcal{D}(H_P) &=\left\{\psi\in \mathscr{H}^2(0,1)\,\left|\,
\begin{array}{c}
 \psi(0)=\psi(1) \\
 \psi'(0)=\psi'(1)
\end{array}\! \right.\right\}\,.
\end{align*}
As customary,  $C^\infty_0(0,1)$ denotes the subspace of $[0,1]\to\mathbb{C}$ functions  that are differentiable -- in the classical sense -- infinitely many times and such that their support is a compact subset of the interval $(0,1)$. In particular, any element in $C^\infty_0(0,1)$ vanishes in a neighbourhood of $x=0$ and of $x=1$.

All such operators have dense domain, with $\mathcal{D}(H_\circ)\subset \mathcal{D}(H_D)\cap  \mathcal{D}(H_P)$. Moreover, integration by parts shows that they are all hermitian.

As a matter of fact, $H_D$ and $H_P$ are in addition self-adjoint on $L^2(0,1)$, whereas $H_\circ$ is not: this can be seen in various ways, an elementary albeit tedious one is to work out the definition \eqref{eq:def-adjoint} for the adjoint and check that indeed $H_D=H_D^\dagger$ and $H_P=H_P^\dagger$, whereas $\mathcal{D}(H_\circ^\dagger)=\mathscr{H}^2(0,1)\varsupsetneq\mathcal{D}(H_\circ)$ (or, more efficiently, one can apply the standard criterion of self-adjointness \cite[Theorem VIII.3]{rs1}, as we tacitly do in some of the mathematical proofs of Sect.~\ref{sec:math-proofs} below).

Moreover, as well-known to physicists, both $H_D$ and $H_P$ have only discrete spectrum (whereas evidently there is no non-zero solution $\psi\in\mathcal{D}(H_\circ)$ to the eigenvalue problem $H_\circ\psi=E\psi$):
\begin{itemize}
 \item[(\emph{D})] the eigenvalues of $H_D$ are the numbers $E^{(D)}_n=\pi^2 n^2$, $n\in\mathbb{N}$, each of which is non-degenerate and with normalised eigenfunction $\psi^{(D)}_n(x)=\sqrt{2}\sin\pi n x$;
 \item[(\emph{P})] the eigenvalues of $H_P$ are the numbers $E^{(P)}_n=4\pi^2 n^2$, $n\in\mathbb{N}_0$, each of which is double-degenerate apart from the non-degenerate ground state $n=0$, and with normalised eigenfunctions $\psi^{(P)}_n(x)=\sqrt{2}\sin 2\pi n x$, $\psi^{(P)}_{-n}(x)=\sqrt{2}\cos 2\pi n x$\,.
\end{itemize}

% . To be precise, let us distinguish once again the \emph{formal action} $\psi\stackrel{H}{\longmapsto} -\psi''$ from the actual definition of $H$ as an operator on $L^2(0,1)$. 
% The point here is (as in very many analogous circumstances for quantum mechanical observables in first quantisation) that the formal action $\psi\stackrel{H}{\longmapsto} -\psi''$ can be realised ``self-adjointly'' in infinite different ways, depending on the choice of the domain, in practice depending on the choice of suitable boundary conditions at $x=0$ and $x=1$. We shall see in a moment that different realisations drive distinct Schr\"{o}dinger dynamics, producing in general distinct evoluted states at $t>0$ starting from the same initial datum at $t=0$.

Correspondingly, let us consider the initial value problem \eqref{eq:SchrIVP} for the ``free Schr\"{o}dinger dynamics'', namely
\begin{equation}\label{eq:SchrIVP-01}
  \begin{cases}
  \ii\frac{\partial}{\partial t}\psi(t,x)\,=\,-\frac{\partial^2}{\partial x^2}\psi(t,x) \\
  \quad \,\psi(0,x)\,=\,\psi_0(x)\qquad\qquad t\in\mathbb{R}\,,\; x\in(0,1)
 \end{cases}
\end{equation}
imposing $\psi(t,\cdot)\in L^2(\mathbb{R})$ at any time $t$. As initial datum, let us pick some $\psi_0\in C^\infty_0(0,1)$. At this stage  \eqref{eq:SchrIVP-01} is to be regarded as a partial differential equation where the differential action $-\frac{\ud^2}{\ud x^2}$ is not characterised as a self-adjoint operator. Thus, \eqref{eq:SchrIVP-01} has the same differential form irrespectively of whether one is considering the Schr\"{o}dinger evolution governed by one or another of the quantum observables $H_\circ,H_D,H_P$ (in fact only the last two are admissible observables, as we shall conclude in a moment), even though the differential equation is the same in all three cases.

Now, interpreting $\psi_0\in \mathcal{D}(H_D)$ (respectively, $\psi_0\in \mathcal{D}(H_P)$), Theorem \ref{thm:selfadj-SchrEq} applied to $H_D$ (resp., to $H_P$) ensures the existence of a unique solution $\psi_D(t,x)$  (resp., $\psi_P(t,x)$) to \eqref{eq:SchrIVP-01} evolving in $\mathcal{D}(H_D)$ and therefore preserving the Dirichlet boundary conditions (resp., evolving in $\mathcal{D}(H_P)$ and therefore preserving the periodic boundary conditions). Thanks to the fact that the eigenfunctions of either operator constitute an orthonormal basis of $L^2(0,1)$, one can determine $\psi_D$ and $\psi_P$ at later times by decomposition along such bases (which is an indirect way to also compute the propagators $e^{-\ii t H_D}$ and $e^{-\ii t H_P}$). Explicitly, from the decompositions
\begin{align*}
 \psi_0\,&=\,\sum_{n\in\mathbb{N}} \alpha_n\psi^{(D)}_n & \alpha_n&:=\langle\psi^{(D)}_n,\psi_0\rangle_{L^2}\,, \\
  \psi_0\,&=\,\sum_{n\in\mathbb{N}_0} \big(\beta_n\psi^{(P)}_n+\beta_{-n}\psi^{(P)}_{-n}\big)  &\beta_{\pm n}&:=\langle\psi^{(D)}_{\pm n},\psi_0\rangle_{L^2}\,,
\end{align*}
one finds
\[
 \begin{split}
  \psi_D(t,x)\,&=\,\sqrt{2}\sum_{n\in\mathbb{N}} e^{-\ii t \pi^2 n^2}\alpha_n\sin\pi n x\,, \\
  \psi_P(t,x)\,&=\,\sqrt{2}\sum_{n\in\mathbb{N}_0} e^{-4\ii t \pi^2 n^2} \big(\beta_n\sin 2\pi n x+\beta_{-n}\cos 2\pi n x\big)\,.
 \end{split}
\]

Clearly $\psi_D\neq\psi_P$: they are distinct Schr\"{o}dinger evolutions from the same initial $\psi_0$. One thus comes to the following conclusions.
\begin{itemize}
 \item[1.] When the actual declaration of self-adjointness of the Hamiltonian ``$H=-\frac{\ud^2}{\ud x^2}$'' is lacking, \eqref{eq:SchrIVP-01} has at least \emph{two} distinct $L^2$-solutions $\psi_D(t,x)$ and $\psi_P(t,x)$ -- in fact, \emph{infinitely many} distinct $L^2$-solutions, one for each of the infinitely many self-adjoint realisations of $H$ in $L^2(0,1)$, each corresponding to a boundary condition of self-adjointness, the general classification of which can be found, e.g., \cite[Example 14.10]{schmu_unbdd_sa}.
 \item[2.] The formal hermiticity of $H$ alone cannot decide the ``physical'' solution, because the boundary conditions of self-adjointness that define the physics of $H$ are not part of the differential equation \eqref{eq:SchrIVP-01}. Formal hermiticity leaves the Schr\"{o}dinger dynamics associated with $H$ with an unphysical infinity of distinct $L^2$-solutions from the same initial datum $\psi_0$.
 \item[3.] The quest for a solution to \eqref{eq:SchrIVP-01} driven by the hermitian-only $H_\circ$ is vain: no non-zero $\psi(t,x)$ solves \eqref{eq:SchrIVP-01} whose spatial support remains a compact in $(0,1)$, owing to the dispersive character of the elliptic operator $-\frac{\ud^2}{\ud x^2}$. Observe, in particular, that the above-considered $\psi_D$ and $\psi_P$ immediately leave the domain $\mathcal{D}(H_\circ)$ as soon as $t>0$. This too is in agreement with Theorem \ref{thm:selfadj-SchrEq}: $H_\circ$ fails to satisfy condition (ii) therein, and consistently there is no solution to $\ii\frac{\ud}{\ud t}\psi(t)=H_\circ\psi(t)$ with $\psi(t)\in\mathcal{D}(H_\circ)$ for every $t>0$.
 %(noticeably, if it existed, it would be unique: see Remark \ref{rem:uniqueness_with_hermitian} below).
\end{itemize}

The same line of reasoning applies to Schr\"{o}dinger Hamiltonians on $L^2(\mathbb{R}^d)$ when the formal action $H=-\Delta+V(x)$, for given real-valued potential $V$, admits a multiplicity of distinct self-adjoint realisations that depend on suitable boundary conditions of self-adjointness at the singularity points of $V$ and at spatial infinity.

An example is the Hydrogenoid Hamiltonian in three dimensions,
\[
 H\,=\,-\frac{\:\hbar^2}{2m}\Delta-\frac{Z e^2}{|x|}
\]
(with all physical constants temporarily reinstated here on purpose). Such $H$ is hermitian with respect to $L^2(\mathbb{R}^3)$ at least on the domain of infinitely-differentiable functions whose support is compact in $\mathbb{R}^3$ and is separated from the origin: on such functions integration by parts indeed yields the symmetry condition \eqref{eq:defsa1}. However, one can show that $H$ admits an infinity of distinct self-adjoint realisations beside the ``ordinary'' Coulomb Hamiltonian studied since the early days of quantum mechanics: each of them is determined by suitable boundary conditions as $|x|\to 0$ which describe an additional interaction localised at the origin, beside the Coulomb attraction (see, e.g., \cite{GM-hydrogenoid-2018} and the precursor results cited therein).

Arguing as before, one sees that the Schr\"{o}dinger equation
\[
 \ii\hbar\frac{\partial}{\partial t}\psi(t,x)\,=\,-\frac{\:\hbar^2}{2m}(\Delta\psi)(t,x)-\frac{Z e^2}{|x|}\psi(t,x)
\]
admits an infinity of distinct $L^2$-solutions starting at time $t=0$ from the same initial datum $\psi_0$ taken in the above-mentioned domain of hermiticity. Only the declaration of the precise domain of self-adjointness for $H$ resolve such an unphysical ambiguity, in which case the unique solution finally becomes $e^{-\ii t H/\hbar}\psi_0$.

\subsection{Self-adjointness of generic quantum observables}\label{sec:SA-generic-obs}~

Not all quantum observables have the role of Hamiltonians in the sense of generators of the dynamics. Thus, strictly speaking, the reasonings of Subsect.~\ref{sec:selfadj-spectralthm}-\ref{sec:non-uniquedynamics} do not provide physical grounds to the self-adjointness of \emph{generic} quantum observables.

Consider, for instance, a system consisting of a free Schr\"{o}dinger particle in one dimension, with Hilbert space $\cH=L^2(\mathbb{R})$, where the observables `free Hamiltonian' $H$, `position' $Q$, and `momentum' $P$ are respectively 
\begin{align}\label{eq:HQP}
H&=-\frac{\ud^2}{\ud x^2}           &  \mathcal{D}(H) &=\mathscr{H}^2(\mathbb{R})     \nonumber \\
Q&=\textrm{multiplication by $x$}           &  \mathcal{D}(Q) &=\{\psi\in L^2(\mathbb{R})\,|\,x\psi\in L^2(\mathbb{R})\}\\
P&=-\ii\frac{\ud}{\ud x}           &  \mathcal{D}(P) &=\mathscr{H}^1(\mathbb{R})\,. \nonumber
\end{align}
All those listed above can be shown to be domains of self-adjointness for the respective operators. The Hamiltonian $H$ governs the quantum dynamics through the Schr\"{o}dinger equation $\ii\frac{\partial}{\partial t}\psi(t,x)=-\frac{\partial^2}{\partial x^2}\psi(t,x)$, and we discussed already how the mathematical well-posedness and physical meaningfulness of such equation impose $H$ to be self-adjoint. But on which grounds does one require position and momentum to be self-adjointly realised as above?

One may argue, from a more abstract point of view, that position and momentum being the generators of the (one-parameter, strongly continuous unitary groups forming the) Weyl $C^*$-algebra, their unique representation on $L^2(\mathbb{R})$ is given by the above self-adjoint $Q$ and $P$ respectively \cite[Chapter 3]{Strocchi-MathQM}. But this has to do with the special role of the position and momentum observables and does not apply to generic quantum observables.

Among those, there are observables of the form $f(H)$, or $f(P)$, or $f(Q)$ for some relevant function $f:\mathbb{R}\to\mathbb{R}$, and for this class of observables one would argue that the self-adjointness of $A$ allows one to construct, by means of the functional calculus, the self-adjoint observable $f(A)$.

Yet, this does not cover all conceivable quantum observables, and additional arguments should be brought in, which can translate into self-adjointness -- and not mere hermiticity -- certain requirements dictated by physical reasoning.

In a sense, one might content oneself to make this claim: as the generator of the dynamics is to be self-adjoint in order to produce the correct quantum evolution (in the sense of Subsect.~\ref{sec:selfadj-spectralthm}), so must be any other quantum observable, \emph{by analogy} with the Hamiltonian.

In the following we shall supplement the latter conclusion by discussing other physically grounded motivations to self-adjointness for quantum observables, in contexts that are not necessarily those of Hamiltonians.

%At the core of such a construction is the approximation of $f(H)$ by means of finite linear combinations of orthogonal projections (namely everywhere defined and bounded operators $E$ on $\cH$ satisfying $E=E^\dagger=E^2$) associated with $H$

\subsection{Self-adjointness of quantum observables with an orthonormal basis of eigenvectors}\label{sec:OBN-estates}~

This is a point of view that is very central in the mathematical structure of quantum mechanics -- one can refer directly to the subtle reasoning presented by Dirac in \cite[Sect.~10]{Dirac-PrinciplesQM} (for a more recent exposition one can see \cite[Sect.~3.3]{weinberg_2015}).

When discussing the measurement mechanism of a quantum observable $A$, one comes to infer that the measurement performed with the system in a particular state makes it jump onto an eigenstate of $A$, the result of the measurement being the corresponding eigenvalue; moreover, the original state must be dependent on such eigenstates, in the sense of being expressed by a linear combination of them. Merging this with the operational assumption that the measurement can be performed on \emph{any} state, one should conclude, as written for instance by Dirac in \cite[Sect.~10]{Dirac-PrinciplesQM}, that the eigenstates of $A$ must form a complete set in the Hilbert space, and this provides an additional constraint on those operators that represent a quantum observable.

This is a paradigmatic line of reasoning in physical introductions to quantum mechanics. In more precise terms, when the underlying Hilbert space $\cH$ is infinite-dimensional and the observable $A$ is (hermitian and) possibly unbounded, the physical request is formulated as follows. If on \emph{any} state $\psi$ of $\mathcal{D}(A)$ it is possible to perform a measurement of $A$ in the sense of a Stern-Gerlach-like experiment, namely a filter with output given by one of the eigenstates of $A$, and hence the generic initial $\psi$ must admit an expansion in eigenstates, then owing to the density of the domain $\mathcal{D}(A)$ the eigenstates of $A$ must constitute an orthonormal basis of $\cH$.

% and therefore its domain $\mathcal{D}(A)$ is a proper dense subspace of $\cH$, the physical request is that if for the measurement of $A$ on any vector
% 
% 
% , but first of all it should be mathematically refined, owing to certain technical subtleties. To be precise, when the underlying Hilbert space is infinite-dimensional, one should mean that the original state on which the measurement is performed can be \emph{approximated, up to an arbitrarily small error}, by a \emph{finite} linear combination of eigenstates of $A$: prescribing instead that any state can be \emph{exactly} expressed by a finite linear combination of eigenstates would imply that $A$ has an algebraic basis of eigenstates -- an occurrence that eventually is incompatible with self-adjointness, owing to the spectral theorem. Furthermore, in terms that account for the possible unboundedness of a densely defined and hermitian $A$, one should rephrase the above reasoning as follows: since the quantum measurement can be performed on an arbitrary $\psi\in\mathcal{D}(A)$ and since $\mathcal{D}(A)$ is dense in $\cH$, then requiring that $\psi$ be approximable by a finite linear combination of eigenstates of $A$ leads to the conclusion that $\mathcal{D}(A)$ must contain an orthonormal basis of eigenvectors.

Clearly this applies to quantum observables for which the measurement is feasible in the ``filter'' sense of a Stern-Gerlach apparatus, like the usual harmonic oscillator, or the free kinetic operators $H_D$ and $H_P$ with Dirichlet or periodic boundary conditions considered in Subsect.~\ref{sec:non-uniquedynamics} (instead, observables like \eqref{eq:HQP} are not included: they have no eigenstates at all).
% 
% -- but many do: we have already recalled that for a free quantum particle in an infinite one-dimensional well, say, the interval $[0,1]$, thus with Hilbert space $\cH=L^2(0,1)$, the self-adjoint Hamiltonian
% \begin{equation*}%\label{eq:DSFlaplacian}
%  H\,=\,-\frac{\ud^2}{\ud x^2}\,,\qquad \mathcal{D}(H)\,=\,\big\{ \psi\in \mathscr{H}^2(0,1)\,\big|\, \psi(0)=0=\psi(1) \big\}
% \end{equation*}
% indeeds admits the orthonormal basis $\{\sqrt{2}\sin n\pi x\,|\,n\in\mathbb{N}\}$ of eigenvectors entirely contained in $\mathcal{D}(H)$. 
For such observables, the above physical reasoning produces a requirement that does characterise them as self-adjoint, and not merely hermitian, operators. The precise formulation is the following -- let us postpone the proof to Sect.~\ref{sec:math-proofs}.

\begin{theorem}\label{thm:onbEV-sa}
 Let $\cH$ be a complex Hilbert space and let $A$ be a \emph{closed hermitian} operator acting on $\cH$ with domain $\mathcal{D}(A)$. Assume that $\mathcal{D}(A)$ contains an orthonormal basis of $\cH$ whose elements are eigenvectors of $A$. Then $A$ is self-adjoint.
\end{theorem}

\begin{example}\label{ex:qao}
 With respect to the Hilbert space $\cH=L^2(\mathbb{R})$ consider the Hamiltonian $H$ of the one-dimensional quantum harmonic oscillator defined as
 \[
  \begin{split}
   \mathcal{D}(H)\;&=\;\Big\{\psi\in L^2(\mathbb{R})\,\Big|\,\int_\mathbb{R}|\psi''(x)|^2\,\ud x\,<\,+\infty\,,\int_\mathbb{R}|x^2\psi(x)|^2\,\ud x\,<\,+\infty\Big\} \\
   H\psi\;&=\;-\psi''+x^2\psi\,.
  \end{split}
  %H\,=\,-\frac{\ud^2}{\ud x^2}+x^2
 \]
 $H$ is a densely defined, unbounded, hermitian operator. The fact that the eigenstates of $H$ form the collection $(N_n e^{x^2/2}\frac{\ud^n}{\ud x^2}e^{-x^2})_{n\in\mathbb{N}_0}$ (the family of `Hermite functions'), where the $N_n$'s are normalisation factors, is familiar from any typical introductory exposition to quantum mechanics, and additionally it is a standard functional-analytic fact that such a collection is indeed an orthonormal basis of $L^2(\mathbb{R})$ -- the so-called basis of Hermite functions. Let us check that $H$ is also a closed operator, so that Theorem \ref{thm:onbEV-sa} establishes that $H$ is self-adjoint. To this aim, let $(\psi_n)_{n\in\mathbb{N}}$ be a sequence in $\mathcal{D}(H)$ and let $\psi,\phi\in L^2(\mathbb{R})$ such that $\psi_n\xrightarrow[]{(L^2)}\psi$ and $-\psi_n''+x^2\psi_n\xrightarrow[]{(L^2)}\phi$ as $n\to\infty$. In view of the definition given in Subsect.~\ref{sec:closedoperators}, one has to show that $\psi\in\mathcal{D}(H)$ and $H\psi=\phi$. From the assumptions one deduces that for every $\eta\in \mathcal{S}(\mathbb{R})$ (the subspace of Schwartz functions) one has
 \[
  \langle\eta,\phi\rangle_{L^2}\,=\,\lim_{n\to\infty}\langle\eta,-\psi_n''+x^2\psi_n\rangle_{L^2}\,=\,\lim_{n\to\infty}\langle-\eta''+x^2\eta,\psi_n\rangle_{L^2}\,=\,\langle-\eta''+x^2\eta,\psi\rangle_{L^2}\,.
 \]
 This means that the \emph{distributions} $\phi$ and $-\psi''+x^2\psi$ are actually the same, whence $-\psi''+x^2\psi=\phi\in L^2(\mathbb{R})$. Thus, indeed, $\psi\in\mathcal{D}(H)$ and $H\psi=-\psi''+x^2\psi=\phi$. The operator $H$ is self-adjoint. If one had chosen a \emph{smaller} domain, say, $\mathcal{S}(\mathbb{R})$, the resulting operator would have been hermitian but \emph{not} self-adjoint (still having the above orthonormal basis of eigenstate in its domain $\mathcal{S}(\mathbb{R})$!), because it would not be closed (and self-adjoint operators are closed -- see Lemma \ref{lem:Aselfadj-is-closed} above). The lack of closedness in this case would emerge by repeating the same check as above: from the conclusion that $-\psi''+x^2\psi=\phi\in L^2(\mathbb{R})$ one could not deduce that $\psi\in\mathcal{S}(\mathbb{R})$. 
\end{example}

% 
% 
% \begin{remark}
%  The actual scope of Theorem \ref{thm:onbEV-sa} includes quantum observables with only an amount of eigenstates that do not span the whole $\cH$ -- a typical example is a Schr\"{o}dinger operator $H=-\frac{\ud^2}{\ud x^2}+V$ on $\cH=L^2(\mathbb{R})$ where $V(x)$ is a finite potential well, thus with a finite number of bound states. For any such observable $A$, meant to be a 
%  
%  
% \end{remark}

\subsection{Self-adjointness for stable observables when expectations define a closed quadratic form}\label{sec:closedsemibdd-sa}~

This is yet another fundamental motivation to self-adjointness of quantum observables. It actually applies to those observables whose expectations are lower semi-bounded, in the sense that 
\begin{equation}\label{eq:lowersemibdd}
 \inf_{\psi\in\mathcal{D}(A)\setminus\{0\}}\frac{\langle\psi, A\psi\rangle}{\|\psi\|^2}\;>\;-\infty\,.
\end{equation}
Such a feature is most relevant when the considered observable is a quantum Hamiltonian, in which case \eqref{eq:lowersemibdd} encodes the \emph{stability} of the quantum system, namely it prevents the existence of a sequence of normalised states attaining arbitrarily negative energy expectations.

\begin{example}\label{ex:semibddforms}~
 \begin{itemize}
  \item[(i)] (Free quantum particle in an infinite one-dimensional well with repulsive boundaries.) With respect to $\cH=L^2[0,1]$ let 
\begin{equation*}%\label{eq:DSFlaplacian}
 H=-\frac{\ud^2}{\ud x^2}\,,\qquad \mathcal{D}(H)\,=\,\big\{ \psi\in \mathscr{H}^2(0,1)\,\big|\, \psi(0)=0=\psi(1) \big\}\,.
\end{equation*}
 The expectations of $H$ are expressed (in the sense of Subsect.~\ref{sec:op-form-domain}) by the energy quadratic form
 \[
  \mathcal{E}_H[\psi]\,=\,\int_0^1|\psi'(x)|^2\ud x\,,\qquad\mathcal{D}[H]\,=\,\big\{ \psi\in \mathscr{H}^1(0,1)\,\big|\, \psi(0)=0=\psi(1) \big\}\,.
 \]
 Poincar\'e's inequality $\int_0^1|\psi'(x)|^2\ud x\geqslant\pi^2\int_0^1|\psi(x)|^2\ud x$ shows that the quadratic form $\mathcal{E}_H[\psi]$ is lower semi-bounded:
 \[
  \mathcal{E}_H[\psi]\,\geqslant\,\pi^2\|\psi\|^2\qquad\forall\psi\in\mathcal{D}[H]\,,
 \]
 and so is therefore the Hamiltonian $H$. 
 \item[(ii)] (Hydrogenoid atoms.) With respect to $\cH=L^2(\mathbb{R}^3)$ consider the energy quadratic form of the Hydrogenoid Hamiltonian $H=-\Delta-\frac{Z}{|x|}$ with $Z>0$, that is,
 \[
  \begin{split}
    \mathcal{E}_H[\psi]\,&=\,\int_{\mathbb{R}^3}|\nabla\psi(x)|^2\,\ud x-Z\int_{\mathbb{R}^3}\frac{\;|\psi(x)|^2}{|x|}\,\ud x \\
    \mathcal{D}[H]\,&=\,\left\{\psi\in L^2(\mathbb{R}^3)\,\Big|\,\int_{\mathbb{R}^3}|\nabla\psi(x)|^2<+\infty\,,\int_{\mathbb{R}^3}\frac{\;|\psi(x)|^2}{|x|}\,\ud x<+\infty\right\}\,.
  \end{split}
 \]
 Based on standard functional-analytic arguments one can prove the inequality $\int_{\mathbb{R}^3}\frac{\;|\psi(x)|^2}{|x|}\ud x\leqslant\|\nabla\psi\|\,\|\psi\|$ $\forall\psi\in\mathcal{D}[H]$, whence
 \[
  \mathcal{E}_H[\psi]\,\geqslant\,\|\nabla\psi\|^2-Z\,\|\nabla\psi\|\,\|\psi\|\,=\,\Big(\|\nabla\psi\|-\frac{Z}{2}\,\|\psi\|\Big)^2-\frac{\,Z^2}{4}\,\|\psi\|^2\,\geqslant\,-\frac{\,Z^2}{4}\,\|\psi\|^2.
 \]
 Therefore, the quadratic form $\mathcal{E}_H[\psi]$ is lower semi-bounded:
 \[
  \mathcal{E}_H[\psi]\,\geqslant\,-\frac{\,Z^2}{4}\,\|\psi\|^2\qquad\forall\psi\in\mathcal{D}[H]\,.
 \]
 \item[(iii)] (Position operator on the real line.) With respect to $\cH=L^2(\mathbb{R})$ consider the quadratic form of the position operator $Q$ (multiplication by $x$), that is,
 \[
  \mathcal{E}_Q[\psi]\,=\,\int_\mathbb{R} x|\psi(x)|^2\,\ud x\,,\quad\mathcal{D}[Q]\,=\,\Big\{\psi\in L^2(\mathbb{R})\textrm{ with }\Big|\int_\mathbb{R} x|\psi(x)|^2\,\ud x\Big|<+\infty\Big\}\,.
 \]
  Along the sequence $(\psi_n)_{n\in\mathbb{N}}$ in $\mathcal{D}[Q]$ defined by $\psi_n:=\mathbf{1}_{[-n-1,-n]}$ (characteristic function of the interval $[-n-1,-n]$) one has $\|\psi_n\|=1$ and $\mathcal{E}_Q[\psi_n]\leqslant -n$. This shows that the quadratic form $\mathcal{E}_Q$ is \emph{not} lower semi-bounded.
 \end{itemize}
\end{example}

As in the present context we are dealing with expectations of certain observables, it is more natural to switch to the quadratic form language.

Modelled on the map $\psi\mapsto\langle\psi,A\psi\rangle$, expectations of a quantum observable behave like a map $\psi\mapsto\mathcal{E}[\psi]$, where $\mathcal{E}[\cdot]$ is a `quadratic form' on an abstract complex Hilbert space $\cH$. By this one means a map $\mathcal{E}:\mathcal{D}[\mathcal{E}]\to\mathbb{C}$, with domain given by the subspace $\mathcal{D}[\mathcal{E}]\subset\cH$, such that the evaluation on a generic $\psi\in\mathcal{D}[\mathcal{E}]$ satisfies $\mathcal{E}[\psi]=\mathcal{E}[\psi,\psi]$, where $\mathcal{E}[\cdot,\cdot]$ is a `sesquilinear form' $\mathcal{D}[\mathcal{E}]\times\mathcal{D}[\mathcal{E}]\to\mathbb{C}$ (that is, anti-linear in the first, linear in the second). Conversely, the sesquilinear form is recovered by the associated quadratic form by means of the `polarisation identity'
\[
 \mathcal{E}[\psi,\phi]\,=\,\frac{1}{4}\Big(\mathcal{E}[\psi+\phi]-\mathcal{E}[\psi-\phi]+\ii\,\mathcal{E}[\psi+\ii\phi]-\ii\,\mathcal{E}[\psi-\ii\phi]\Big)\,.
\]

Quadratic/sesquilinear forms of relevance in quantum mechanics must satisfy $\mathcal{E}[\psi,\phi]=\overline{\mathcal{E}[\phi,\psi]}$ $\forall\psi,\phi\in\mathcal{D}[\mathcal{E}]$, in order for the expectations $\mathcal{E}[\psi]$ to be \emph{real}, in which case $\mathcal{E}[\cdot]$ is called a `symmetric' form. A symmetric form is lower semi-bounded when, for some $m\in\mathbb{R}$, $\mathcal{E}[\psi]\geqslant m\|\psi\|^2$ $\forall \psi\in \mathcal{D}[\mathcal{E}]$.

 When the quadratic form $\mathcal{E}[\cdot]$ is densely defined (i.e., when the subspace $\mathcal{D}[\mathcal{E}]$ is dense in $\cH$), it identifies a linear operator $A$ associated with the form, by means of the natural definition
 \begin{equation}\label{eq:fromFormToOp}
  \begin{split}
   \mathcal{D}(A)\;&:=\;\big\{\psi\in\mathcal{D}[\mathcal{E}]\,|\,\exists\,\xi_\psi\in\cH\textrm{ such that }\langle\phi,\xi_\psi\rangle=\mathcal{E}[\phi,\psi]\;\forall\phi\in\mathcal{D}[\mathcal{E}]\big\} \\
   A\psi\;&:=\;\xi_\psi\,.
  \end{split}
 \end{equation}
 Here $\xi_\psi$, if it exists, is necessarily unique, owing to the density of $\mathcal{D}[\mathcal{E}]$, which makes the definition \eqref{eq:fromFormToOp} indeed unambiguous. With no further assumption on the quadratic form $\mathcal{E}[\cdot]$, \eqref{eq:fromFormToOp} may produce an operator $A$ whose domain is not even dense in $\cH$ (but in the physically relevant case of Example \ref{ex:fromformtoposition} below $\mathcal{D}(A)$ \emph{is} dense): in a moment we shall see when the density of the operator domain is guaranteed.

 \begin{example}\label{ex:fromformtoposition}
  The quadratic form $\mathcal{E}_Q$ from Example \ref{ex:semibddforms}(iii) identifies, by means of \eqref{eq:fromFormToOp}, the operator $Q$ given by $Q\psi=x\psi$, $\mathcal{D}(Q)=\{\psi\in L^2(\mathbb{R})\,|\,x\psi\in L^2(\mathbb{R})\}$.
%   \[
%    Q\psi\,=\,x\psi\,,\qquad\mathcal{D}(Q)\,=\,\{\psi\in L^2(\mathbb{R}\,|\,x\psi\in L^2(\mathbb{R}\}\,.
%   \]
 \end{example}

 Now, if one regards $\mathcal{E}[\psi]$ as an average of measurements of some quantum observable, performed on the state $\psi$ of the system, a frequent occurrence is: expectations are lower semi-bounded (there exists a lower bound $m\in\mathbb{R}$ with $\mathcal{E}[\psi]\geqslant m\|\psi\|^2$ $\forall \psi\in \mathcal{D}[\mathcal{E}]$), and some special form of weak continuity holds.

 Imagine for concreteness a measurement apparatus that is fixed (like a counter saying whether a particle crosses a given spatial region at a given time, and with which energy) and a quantum system that is slightly perturbed (e.g., shifted) so as to occupy a sequence $(\psi_n)_{n\in\mathbb{N}}$ of very close states on which $\mathcal{E}$ is evaluated, namely states belonging to $\mathcal{D}[\mathcal{E}]$, up to a final configuration and hence a final state $\psi\in\cH$ with $\|\psi_n-\psi\|\to 0$ as $n\to\infty$. Assume furthermore that on the difference $\psi_n-\psi_m$, as $n$ and $m$ becomes larger and larger, the evaluation of $\mathcal{E}$ becomes smaller and smaller, that is, $\mathcal{E}[\psi_n-\psi_m]\to 0$ as $n,m\to\infty$. In such a circumstance one would like to be able to conclude, as natural as it appears, that on the limit state $\psi$ too it is possible to evaluate the expectation $\mathcal{E}[\psi]$, i.e., $\psi\in\mathcal{D}[\mathcal{E}]$, and that indeed $\mathcal{E}[\psi_n-\psi]\to 0$ as $n\to\infty$. A quadratic form with such a behaviour is said to be closed.

 Thus, a lower semi-bounded form $\mathcal{E}$ is `closed' when, if $(\psi_n)_{n\in\mathbb{N}}$ is a sequence in $\mathcal{D}[\mathcal{E}]$ with $\psi_n\to\psi$ for some $\psi\in\cH$ and with $\mathcal{E}[\psi_n-\psi_m]\to 0$ as $n,m\to\infty$, then $\psi\in\mathcal{D}[\mathcal{E}]$ and $\mathcal{E}[\psi_n-\psi]\to 0$ as $n\to\infty$.

 \begin{example}\label{ex:form-closed-nonclosed} Let $\cH=L^2(\mathbb{R})$.
 \begin{itemize}
  \item[(i)] Consider the (non-negative) quadratic form
  \[
   \mathcal{E}[\psi]\,=\,\int_\mathbb{R} x^2|\psi(x)|^2\,\ud x\,,\qquad\mathcal{D}[\mathcal{E}]\,=\,\{\psi\in L^2(\mathbb{R})\,|\,x\psi\in L^2(\mathbb{R})\}\,,
  \]
  namely the quadratic form that expresses the expectation of the observable ``position squared'', i.e., the operator $\psi\mapsto x^2\psi$. Pick a sequence $(\psi_n)_{n\in\mathbb{N}}$ in $\mathcal{D}[\mathcal{E}]$ such that $\psi_n\xrightarrow[]{(L^2)}\psi$ for some $\psi\in L^2(\mathbb{R})$ and $\mathcal{E}[\psi_n-\psi_m]\to 0$ as $n,m\to\infty$. Then $(x\psi_n)_{n\in\mathbb{N}}$ is a sequence in $L^2(\mathbb{R})$ with the Cauchy property $\| x\psi_n-x\psi_m\|^2_{L^2}=\mathcal{E}[\psi_n-\psi_m]\to 0$ and therefore, by completeness, there exists $\phi\in L^2(\mathbb{R})$ such that $x\psi_n\xrightarrow[]{(L^2)}\phi$. As $L^2$-limits are also point-wise (almost everywhere) limits, then $\phi(x)=\lim_{n\to\infty}x\psi_n(x)=x\psi(x)$ for almost every $x\in\mathbb{R}$: thus, $x\psi=\phi\in L^2(\mathbb{R})$, which proves that $\psi\in\mathcal{D}[\mathcal{E}]$. Moreover,  $\mathcal{E}[\psi_n-\psi]=\|x\psi_n-x\psi\|^2_{L^2}\to 0$. The conclusion is that the form $\mathcal{E}$ is closed. 
  \item[(ii)] Consider now the (non-negative) quadratic form $\mathcal{E}[\psi]=|\psi(0)|^2$, $\mathcal{D}[\mathcal{E}]=C(\mathbb{R})\cap L^2(\mathbb{R})$ (the square-integrable functions on $\mathbb{R}$ that are also continuous). Along the sequence $(\psi_n)_{n\in\mathbb{N}}$ in $\mathcal{D}[\mathcal{E}]$ defined by $\psi_n:=e^{-nx^2}$ one has $\|\psi_n\|_{L^2}\to 0$ as $n\to\infty$ and $\mathcal{E}[\psi_n-\psi_m]=|\psi_n(0)-\psi_m(0)|^2=|1-1|^2=0$; however, $\mathcal{E}[\psi_n]=1\to\!\!\!\!\!\!\!{/}\;\;\,0$. Hence this form is not closed.
 \end{itemize} 
 \end{example}

 The crucial point in this context is: requiring that the expectations of an observable behave as a densely defined, lower semi-bounded, closed quadratic form, forces the underlying linear operator \eqref{eq:fromFormToOp} associated with the form to be self-adjoint (and not merely hermitian). Indeed, one has the following.

 \begin{theorem}\label{thm:closed-sa}
  Let $\cH$ be a complex Hilbert space and let $\mathcal{E}$ be a quadratic form on $\cH$ such that the domain $\mathcal{D}[\mathcal{E}]$ is a dense subspace of $\cH$ and that the form is lower semi-bounded and closed. Then the operator $A$ associated with $\mathcal{E}$ through definition \eqref{eq:fromFormToOp} is self-adjoint, and $\mathcal{E}[\psi]=\langle\psi,A\psi\rangle$ for every $\psi\in\mathcal{D}(A)$. 
 \end{theorem}

 Let us postpone the proof to Section \ref{sec:math-proofs}.
 
 %The catch from Theorem \ref{thm:closed-sa} ****

 For completeness of presentation, it is instructive to also mention an alternative, equivalent characterisation of lower semi-bounded, closed quadratic form, because it too has quite a transparent interpretation of how expectations of a quantum observable behave along a convergent sequence of states.

 \begin{proposition}\emph{[See, e.g. \cite[Proposition 10.1]{schmu_unbdd_sa}.]}
   A lower semi-bounded quadratic form $\mathcal{E}$ on a complex Hilbert space $\cH$ is closed \emph{if and only if} the map $\psi\mapsto\mathcal{E}[\psi]$ is lower semi-continuous, meaning that for \emph{every} convergent sequence $(\psi_n)_{n\in\mathbb{N}}$ one has
 \begin{equation}\label{eq:lowersemicont}
  \mathcal{E}\Big[\lim_{n\to\infty}\psi_n\Big]\,\leqslant\,\liminf_{n\to\infty}\mathcal{E}[\psi_n]
 \end{equation}
 (having tacitly extended $\mathcal{E}$ on the whole $\cH$ by setting $\mathcal{E}[\psi]=+\infty$ for all those vectors that do not belong to $\mathcal{D}[\mathcal{E}]$, in order for the above inequality to be consistent).
 \end{proposition}

 For instance, with reference to the quadratic form $\mathcal{E}$ of Example \ref{ex:form-closed-nonclosed}(ii) above, the sequence with $\psi_n(x)=e^{-x^2}(1-e^{-nx^2})$ has the properties $\psi_n\in\mathcal{D}[\mathcal{E}]$, $\mathcal{E}[\psi_n]=0$, $\psi_n\to e^{-x^2}$, $\mathcal{E}[\lim_{n\to\infty}\psi_n]=\mathcal{E}[e^{-x^2}]=1$, and therefore inequality \eqref{eq:lowersemicont} is violated (one would have $1\leqslant 0$).

\subsection{Self-adjointness to ensure the expansion of any state in generalised eigenfunctions of the observable}\label{sec:generalized-EF}~

 As a counterpart and follow-up of Subsect.~\ref{sec:OBN-estates}, let us discuss one last instance where the actual self-adjointness of quantum observables (and not just their mere hermiticity) is crucial for certain fundamental physical requirements of the theory.

 Here we admittedly revert the previous line of reasoning: we assume in the first place that quantum observables are self-adjoint and discuss one physically notable consequence of such assumption. Since the mere assumption of hermiticity would not suffice to deduce such consequence, this discussion is meant to shed further light on the physical role of self-adjointness in quantum mechanics.

 We refer to the possibility, required at various stages of standard physical discussions on the subject (see, e.g., \cite[Sect.~3.10-3.11]{Schiff-QM-1968}), to expand a generic state $\psi\in\cH$ not just as a discrete series of eigenstates, if any, of a given observable $A$, but also as a continuous expansion in terms of new objects that, while not being vectors in $\cH$, yet behave as eigenvectors \emph{of the sole formal action} of $A$. This possibility has relevance both technically, for computational manipulations, and conceptually, for the interpretation given to the coefficients of such expansion. Let us stress that two notions of ``eigenvectors'' are involved now: proper physical states $\psi\in\mathcal{D}(A)$ such that $A\psi=a\psi$ for some $a\in\mathbb{R}$, and generalised objects (in practice: non-square-integrable functions) that only satisfy the eigenvalue/eigenvector problem for the formal action of $A$. The standard concrete example below clarifies the picture.

%  
%  demand that, given an observable $A$ acting on a Hilbert space $\cH$, one can express a generic physical state $\psi\in\cH$, as a linear combination of eigenstates of $A$ (see, e.g., \cite[Sect.~3.10]{Schiff-QM-1968}). Such assertion, that implicitly dates back even to the old quantum theory and Bohr's postulates, is intimately connected with the role that the axioms assign to the eigenfunctions and the eigenvalues of $A$ in the measurement process.
%  
%  The typical physical discussion at this point distinguishes between the ordinary series expansion of $\psi$ over a countable basis of eigenvectors, and a continuous expansion in terms of ``generalised eigenstates''. In fact, a quantum observable may admit an orthonormal basis of eigenstates in $\cH$ as well as no eigenstates at all. 

 \begin{example}\label{ex:expansions-discr-cont}~
 
 \begin{itemize}
  \item[(i)] With respect to $\cH=L^2(\mathbb{R})$, the harmonic oscillator's self-adjoint Hamiltonian previously considered in Example \ref{ex:qao} does have in its domain the orthonormal basis $(\psi_n)_{n\in\mathbb{N}_0}$ of Hermite functions, and consequently a generic $\psi\in L^2(\mathbb{R})$ can be written as
 \begin{equation}\label{eq:discrete-exp}
  \psi\;=\;\sum_{n=0}^\infty c_n \psi_n\,,\qquad c_n\;:=\;\langle\psi_n,\psi\rangle\;=\;\int_{\mathbb{R}}\overline{\psi_n(x)}\,\psi(x)\,\ud x\,.
 \end{equation}
 \item[(ii)] Instead, the self-adjoint momentum operator
 \[
  P\;=\;-\ii\frac{\ud}{\ud x}\,,\qquad\mathcal{D}(P)\;=\;\mathscr{H}^1(\mathbb{R})
 \]
 has no square-integrable eigenfunctions at all. In this case, however, one can represent a generic $\psi\in L^2(\mathbb{R})$ in terms of its Fourier transform $\widehat{\psi}\in L^2(\mathbb{R})$ as
 \begin{equation}\label{eq:Ftransfdef}
  \psi(x)\;=\;\frac{1}{\sqrt{2\pi}}\int_\mathbb{R}e^{\ii x p}\,\widehat{\psi}(p)\,\ud p\,,
 \end{equation}
 which in analogy to \eqref{eq:discrete-exp} can be interpreted as
 \begin{equation}\label{eq:cont-exp}
  \psi(x)\,=\,\int_\mathbb{R}c(p)\,\psi_p(x)\,\ud p\,,\quad \psi_p(x)\,:=\,\frac{e^{\ii p x}}{\sqrt{2\pi}}\,,\quad c(p)\,:=\,\int_\mathbb{R}\overline{\psi_p(x)}\,\psi(x)\,\ud x\,.
 \end{equation}
 The plane waves $\psi_p$, $p\in\mathbb{R}$, have the natural interpretation of ``generalised'' eigenfunctions of $P$, in that $-\ii\frac{\ud}{\ud x}\psi_p=p\psi_p$ but $\psi_p\notin L^2(\mathbb{R})$, and the expansion \eqref{eq:cont-exp} provides a continuous counterpart to \eqref{eq:discrete-exp}. (In \eqref{eq:Ftransfdef} one requires the customary technical caveat that such expression is only well-defined when $\widehat{\psi}$ has good integrability properties, e.g., $\widehat{\psi}\in L^1(\mathbb{R})\cap L^2(\mathbb{R})$, whereas in general it must be understood as an $L^2$-limit of integrals over $[-L,L]$ with $L\to +\infty$. Analogous considerations apply to \eqref{eq:cont-exp}.)
 \end{itemize}
 \end{example}

 Keeping for concreteness the setting $\cH=L^2(\mathbb{R})$ of Example \ref{ex:expansions-discr-cont}, we consider the following physical request for an operator $A$ acting on $\cH$ to be interpreted as quantum observable.

 One requires $A$ to admit an at most countable set of functions $\psi_n\in L^2(\mathbb{R})$ and a continuum of non-$L^2$ functions $\psi_\lambda$ (thus, with discrete label $n$ and continuum label $\lambda$) such that:
 \begin{itemize}
  \item[(i)] each $\psi_n$ is a normalised eigenfunction of $A$, say, $A\psi_n=E_n\psi_n$ with real eigenvalues $E_n$ (the $E_n$'s are said to constitute the `point spectrum' of $A$),
  \item[(ii)] each $\psi_\lambda$ is an eigenfunction of the \emph{formal action} of $A$, say, $A\psi_\lambda=E(\lambda)\psi_\lambda$ with real eigenvalues $E_\lambda$ (the $E_\lambda$'s are said to constitute the `continuous spectrum' of $A$),
  \item[(iii)] and furthermore any $\psi\in L^2(\mathbb{R})$ can be expressed as
 \begin{equation}\label{eq:cont-disc-exp}
  \psi\;=\;\sum_n c_n\psi_n+\int c(\lambda)\,\psi_\lambda\,\ud\lambda
 \end{equation}
 (in the sense of $L^2$-convergent series and integral) with coefficients
 \begin{equation}
  c_n\;=\;\int_{\mathbb{R}}\overline{\psi_n(x)}\,\psi(x)\,\ud x\,,\qquad c(\lambda)\;=\;\int_{\mathbb{R}}\overline{\psi_\lambda(x)}\,\psi(x)\,\ud x\,.
 \end{equation}
 \end{itemize}

 This is interpreted by saying that (for a normalised state $\psi$ of the considered quantum system) $|c_n|^2$ expresses the probability of finding the value $E_n$ of the point spectrum when measuring the observable $A$ on $\psi$, and $|c(\lambda)|^2$ gives the probability density to measure the value $E_\lambda$ of the continuum spectrum.
 %, with $\sum_n|c_n|^2+\int|c(\lambda)|^2\ud\lambda=1$.

 As a matter of fact, the above requirements are matched, and the expansion \eqref{eq:cont-disc-exp} is made precise and consistent, only when the operator $A$ under consideration is self-adjoint with respect to the underlying Hilbert space $\cH$. The sole hermiticity would not be sufficient.

 Rigorous formulations of results of this sort require a beautiful and well established mathematical apparatus (spectral theorem, rigged Hilbert spaces, nuclear spaces, nuclear theorem, theory of generalised functions) that is not part of the minimal background of physicists exposed to early courses of quantum mechanics (and here we are even avoiding dealing with the precise definition of spectrum $\sigma(A)$ of $A$ and its basic properties (see, e.g., \cite[Sect.~VIII.1]{rs1}). We refer to the exhaustive analysis developed in full rigour across the 1950's and 1960's for abstract self-adjoint operators by Gel$'$fand and Vilenkin \cite[Chapt.~I]{Gelfand-Vilenkin-1964}, Gel$'$fand and Shilov \cite[Chapt.~IV]{Gelfand-Shilov-1967}, Berezans$'$ki\u{\i} \cite[Chapt.~5]{Berezanskii-1968}, with also specific results by Browder, Garding, Kac, Povzner (see \cite[Sect.~C.5]{Simon-82-Schroedinger_semigroups} and references therein). We also refer to Berezin and Shubin \cite[Chapt.~S1.2]{Berezin-Shubin-1983} and Simon \cite[Sect.~C.5]{Simon-82-Schroedinger_semigroups} for subsequent discussions, especially for Schr\"{o}dinger self-adjoint operators on $L^2$-space. A somewhat more direct abstract construction was later established by Poerschke, Stolz, and Weidmann \cite{Poerschke-Stolz-Weidmann-1989}.

 For our presentation, we select the following explicit statement, that has the advantage of referring to the much more familiar classes of Schwartz functions $\mathcal{S}(\mathbb{R}^d)$ and distributions $\mathcal{S}'(\mathbb{R}^d)$ (namely the topological dual of $\mathcal{S}(\mathbb{R}^d)$): these are indeed mathematical notions that physicists encounter pretty early in their training.

  \begin{theorem}\label{thm:resized}
  Let $d\in\mathbb{N}$ and let $A$ be a self-adjoint operator with respect to the Hilbert space $L^2(\mathbb{R}^d)$ such that  $\mathcal{S}(\mathbb{R}^d)\subset\mathcal{D}(A)$ and $A\mathcal{S}(\mathbb{R}^d)\subset \mathcal{S}(\mathbb{R}^d)$. Then $A$ admits a collection $(F_\lambda)_{\lambda\in\mathcal{L}}$ of elements of $\mathcal{S}'(\mathbb{R}^d)$ and a collection $(E_\lambda)_{\lambda\in\mathcal{L}}$ of reals, where $\mathcal{L}$ is a possibly continuous index set, such that
    \begin{equation}\label{eq:geneigenfFlambda}
    F_\lambda(A\varphi)\,=\,E_\lambda F_\lambda(\varphi)\qquad \forall\varphi\in \mathcal{S}(\mathbb{R}^d)\,,\quad\forall\lambda\in\mathcal{L}\,,
   \end{equation}
  and 
     \begin{equation}\label{eq:Flambdacompleteness}
    F_\lambda(\varphi)\,=\,0\quad\forall\lambda\in\mathcal{L}\qquad\Rightarrow\qquad \varphi\,\equiv\,0\,.
   \end{equation}
 \end{theorem}

 The above statement is only scratching the surface of a much richer result (see Gel$'$fand and Shilov \cite[Sect.~IV.5, Theorems 1 and 2]{Gelfand-Shilov-1967} or \cite[Sect.~I.4]{Gelfand-Vilenkin-1964}), but at least in this re-sized version it presents a straightforward interpretation: \eqref{eq:geneigenfFlambda} expresses, in the duality sense of distributions, the eigenvalue equation
 \[
  AF_\lambda\,=\,E_\lambda F_\lambda\,, \qquad \lambda\in\mathcal{L}\,,
 \]
 whence the nomenclature of generalised eigenfunctions for $A$ for the $F_\lambda$'s, and \eqref{eq:Flambdacompleteness} expresses their completeness. In Example \ref{ex:expansions-discr-cont}, where $A$ is the one-dimensional self-adjoint momentum observable, the $F_\lambda$'s are precisely the plane waves denoted as $\psi_p$ therein: they are not square-integrable, yet belong to $\mathcal{S}'(\mathbb{R})$. Completeness of the plane waves amounts to saying that if the Fourier transform $\widehat{\varphi}$ of some $\varphi\in\mathcal{S}(\mathbb{R})$ vanishes, so does $\varphi$.

 As just mentioned, various other properties of physical relevance emerge in the setting of Theorem \ref{thm:resized}, except that their rigorous formulation requires a considerable amount of mathematical weaponry. In particular, the generalised eigenvalue $E_\lambda$ do exhaust the spectrum of $A$ in a precise measure-theoretic sense, and moreover, 
 which was the original motivation, it is possible to expand a generic $\varphi\in\mathcal{S}(\mathbb{R}^d)$ as a suitable integral of the generalised eigenfunctions $F_\lambda$, much in the same spirit as the concrete expansion \eqref{eq:cont-exp} and the heuristic expansion \eqref{eq:cont-disc-exp} above (only, for such integral, the measure is \emph{not} necessarily the Lebesgue measure as in \eqref{eq:cont-exp} and need be constructed from the spectral measure of $A$).

\section{The self-adjointness problem in quantum mechanics: \\ a partial retrospective review}\label{sec:selfadjproblem}

 It is fair to claim (see, e.g., \cite[Sect.~VIII.11]{rs1}) that among the mathematical problems in quantum mechanics the self-adjointness problem is conceptually the first that need be settled, prior to embarking on the spectral, dynamical, and scattering analysis of the considered quantum system. It consists of characterising a formal Hamiltonian, whose action is initially dictated by physical reasonings such as first quantisation, as a self-adjoint operator acting on the underlying Hilbert space, thus declaring a valid domain of self-adjointness. And, in those cases where the same formal action can be associated with distinct self-adjoint realisations, the problem is to characterise and investigate the various realisations in order to distinguish among the different physics modelled by each such observables.
 
 We devote this brief Section to a concise list of the most relevant self-adjointness problems that have been solved in the past or are object of current investigation.

 For a vast majority of quantum Hamiltonians of interest, their self-adjointness is indeed already established today, which explains why physicists do not usually have to bother with it. Yet, we believe it is instructive, in the the pedagogical presentation we are proposing, both to make one aware that the proof of self-adjointness of typical quantum Hamiltonians has been a non-trivial task on top of past research agendas, and to stress that the problem is still active today for quantum models of theoretical or applied relevance.

 Here is a (non-exhaustive) list of the most relevant categories of quantum observables in terms of the corresponding self-adjointness problem.

 \textbf{I. Molecular Hamiltonians.} Non-relativistic Hamiltonians for ordinary atoms and molecules were first proved to be self-adjoint by Kato in 1951 \cite{Kato-1951}. Additional references and details are in the notes to \cite[Chapter X]{rs2}, in \cite[Chapter 1]{Cycon-F-K-S-Schroedinger_ops}, and in \cite{Simon-Schr-XX}.
 
 \textbf{II. External magnetic fields.} Self-adjointness of non-relativistic Hamiltonians given by Schr\"{o}dinger operators minimally coupled with an external magnetic vector field: the first general proof (in terms of generality of assumptions) was established by Leinfelder and Simander in 1981 \cite{Leinfelder-Simander-81}. Additional references and details in \cite{Cycon-F-K-S-Schroedinger_ops,Simon-Schr-XX}.

 \textbf{III. Aharonov-Bohm Hamiltonians.} Clearly this is closely connected to the previous point: the self-adjointness problem for Aharonov-Bohm Hamiltonians (thus, with external magnetic field) results in the study of a family of self-adjoint realisations, by now well understood in a variety of settings: details and references in 
 \cite{Dabrowski-Stovicek-1997-AharonovBohm,Adami-Teta-1998-AharonovBohm,Oliveira-Pereira-2008-AharonovBohm}.

 \textbf{IV. Dirac operators.} Self-adjointness proved for one-body free Dirac operators, as well as Dirac operators with external fields, in particular with Coulomb interaction: a research line that became particularly active in the 1970's, until contemporary times. References and details in \cite{Thaller-Dirac-1992,Gallone-AQM2017,MG_DiracCoulomb2017}. The proof of a self-adjoint realisation of the two-body Dirac-Coulomb Hamiltonian in three dimensions is much more recent: \cite{Deckert-Oelker__2DiracCoulomb}.
 
 \textbf{V. Contact interactions.} Self-adjointness problem solved for Hamiltonians of one non-relativistic quantum particle subject to a contact interaction supported at one point or at a collection of distinct points, as well as Hamiltonians of two non-relativistic quantum particles coupled by an interaction of zero range: a complex of investigations mainly developed in the 1970's and early 1980's. References and details in \cite{albeverio-solvable}.
 
 \textbf{VI. Non-relativistic particles on metric graphs.} Hamiltonians for non-relativistic particles constrained on a metric graph (``metric'' here indicating that each edge is isometric to a segment $(0,L)$ for some $L\in\mathbb{R}^+\cup\{+\infty\}$: sometimes this structure is induced by a global metric, as when the graph is embedded in $\mathbb{R}^d$), either moving freely on each edge or possibly subject to additional external interaction.  The original underlying physical model were aromatic hydrocarbon molecules, more modern applications are quantum wires of semiconductors, carbon nanotubes, etc. Their self-adjointness was first characterised, in terms of local boundary conditions at the graph's vertices, by Kostrykin and Schrader \cite{Kostrykin-Schrader-1999}, and by Kurasov and Stenberg \cite{Kurasov-Stenberg-JPA-2002}, in the late 1990's and early 2000's. Further details and references in \cite[Sect.~K.4.2]{albeverio-solvable} and \cite{Kuchment-quantumgraphs-2008}.

 \textbf{VII. Multi-particle systems with zero-range interactions.} The study of Hamiltonians for $N$-body particle systems, $N\geqslant 3$, subject to a two-body interaction with zero-range has occupied a central position in the mathematical physics research agenda since the late 1980's (with precursors in the late 1960's), and is being boosted by current experimental advances in preparing cold atoms with such effective interactions \cite{Braaten-Hammer-2006,Naidon-Endo-Review_Efimov_Physics-2017}. The self-adjointness problem has only been solved in partial cases, depending on the number of particles and the content of bosons and fermions: details and references in \cite{Minlos-1987,Minlos-Shermatov-1989,Kuperin-Makarov-Merk-Motovilov-Pavlov-1989-JMP1990,DFT-1994,CDFMT-2012,CDFMT-2015,MO-2016,MO-2017,M2020-BosonicTrimerZeroRange}.
 %Griesemer-Hofacker-Linden-2019,
  
 \textbf{VIII. Interaction supported on curves or surfaces.} Self-adjointness proved, along a recent mainstream, for non-relativistic or semi-relativistic quantum Hamiltonians for particles subject to an interaction only supported on lower-dimensional sets, such as Landau Hamiltonians with $\delta$-potentials supported on curves, Dirac operators with $\delta$-shell interactions, Schr\"{o}dinger Hamiltonians with interaction supported at surfaces \cite{Behrndt-Langer-Loto-AHP2013,Behrndt-Exner-Holtzmann-Loto-2019-Diracshell,Behrndt-Exner-Holtzmann-Loto-2020-LandauHamilt}

 \textbf{IX. Continuous models of topological quantum phases.} Self-adjointness solved in the recently flourishing subject of topological quantum phases in artificial anon-structures or in bulk crystals \cite{Bernevig-TopIns-2013,Shen2013-TopInsul}, when the Hamiltonian is studied in the continuous limit (field theoretic) representation. In this setting, at a technical level one faces the self-adjointness problem for certain Dirac operators: for a recent application to a topological quantum wire see \cite{Ahari-Ortiz-Seradjeh-SelfAdjTopolQphases2016} and references therein.

\section{Some mathematical proofs}\label{sec:math-proofs}

As announced, we postponed to this Section some mathematical proofs whose essence, if not their entirety, can be presented as instructive side material to a physical audience that already possesses some standard functional-analytic and operator-theoretic tools.

\begin{proof}[Proof of the estimate in Example \ref{Popas-example}]
 Let us simplify for $\ii$ in ``$QP-PQ=\ii\mathbbm{1}$'' and re-interpret it as ``$[X,D]=\mathbbm{1}$'' (``$D$'' here is a reminder for derivative). The assumption then takes the form  
  \[
  \big\|[X,D]-\mathbbm{1}\big\|_{\mathrm{op}}\;\leqslant\;\varepsilon\,.
 \]
 By multiplying $D$ by a constant and dividing $X$ by the same constant, one may normalise $\|X\|_{\mathrm{op}}=\frac{1}{2}$. The estimate to prove takes the form $\|D\|_{\mathrm{op}}\geqslant\log\frac{1}{\varepsilon}$.
 Set $E:=[X,D]-\mathbbm{1}$. Then $\|E\|_{\mathrm{op}}\leqslant\varepsilon$ and, by standard induction,
 \[
  [X,D^n]\;=\;nD^{n-1}+D^{n-1}E+D^{n-2}ED+\cdots+ED^{n-1}\qquad\forall n\in\mathbb{N}\,.
 \]
 The triangular inequality then yields
 \[
  n\|D^{n-1}\|_{\mathrm{op}}\;\leqslant\;\|[X,D^n]\|_{\mathrm{op}}+ n\varepsilon\|D\|^{n-1}\qquad\forall n\in\mathbb{N}\,.
 \]
 We can further estimate $\|[X,D^n]\|_{\mathrm{op}}\leqslant\|D^n\|_{\mathrm{op}}$, having used the triangular inequality again and the fact that $\|X\|_{\mathrm{op}}=\frac{1}{2}$. Therefore,
  \[
  n\|D^{n-1}\|_{\mathrm{op}}\;\leqslant\;\|D^n\|_{\mathrm{op}}+ n\varepsilon\|D\|^{n-1}\qquad\forall n\in\mathbb{N}\,.
 \] 
 Dividing both sides by $n!$ and summing in $n$ we get 
 \[
  \sum_{n=1}^{\infty}\frac{\;\|D^{n-1}\|_{\mathrm{op}}}{(n-1)!}\;\leqslant\;\sum_{n=1}^{\infty}\frac{\;\|D^{n}\|_{\mathrm{op}}}{n!}+\varepsilon\sum_{n=1}^{\infty}\frac{\;\|D^{n-1}\|_{\mathrm{op}}}{(n-1)!}\,,
 \]
 i.e.,
  \[
  \sum_{n=0}^{\infty}\frac{\;\|D^n\|_{\mathrm{op}}}{n!}\;\leqslant\;\sum_{n=1}^{\infty}\frac{\;\|D^{n}\|_{\mathrm{op}}}{n!}+\varepsilon \sum_{n=0}^{\infty}\frac{\;\|D^n\|_{\mathrm{op}}}{n!}\,.
 \]
 This implies
 \[
  1\;\leqslant\;\varepsilon\,e^{\|D\|_{\mathrm{op}}}\,,
 \]
 whence the conclusion $\|D\|_{\mathrm{op}}\geqslant\log\frac{1}{\varepsilon}$. 
\end{proof}

Let us proceed with the proof of Theorem \ref{thm:selfadj-SchrEq}. As mentioned already, one implication is based upon Stone's theorem.

\begin{theorem}[Stone's theorem]\label{thm:Stone} Let $\{U(t)\,|\,t\in\mathbb{R}\}$ be a strongly continuous one-parameter unitary group on a Hilbert space $\cH$. Then there exists a unique self-adjoint operator $A$ on $\cH$ such that $U(t)=e^{-\ii t A}$ for $t\in\mathbb{R}$.  
\end{theorem}

The  statement of Stone's theorem mentions the operator $e^{-\ii t A}$: it is constructed from $A$ by means of of the functional calculus of self-adjoint operators (see, e.g., \cite[Sect.~5.3]{schmu_unbdd_sa}). For the proof of  Stone's theorem we refer, e.g., to \cite[Theorem VIII.8]{rs1} or \cite[Theorem 6.2]{schmu_unbdd_sa}. %\textcolor{red}{$\rightarrow$ to insert nevertheless?}

\begin{proof}[Proof of Theorem \ref{thm:selfadj-SchrEq}]~
 
\underline{Implication (ii) $\Rightarrow$ (i)}. Define $U(t):=e^{-\ii t H}$ by means of the functional calculus of self-adjoint operators. Standard properties of the functional calculus immediately imply that $\{U(t)\,|\,t\in\mathbb{R}\}$ is a strongly continuous one-parameter unitary group, with also $U(t)\mathcal{D}(H)\subset\mathcal{D}(H)$ and  $HU(t)\psi_0=U(t)H\psi_0$ for every $t\in\mathbb{R}$ and $\psi_0\in\mathcal{D}(H)$.

Next, define the auxiliary operator
  \begin{equation*}\tag{*}\label{eq:defA}
   \begin{split}
    \mathcal{D}(A)\;&:=\;\left\{\psi\in\cH\,\left|\,\exists\,\frac{\ud}{\ud t}\Big|_{t=0}U(t)\psi:=\lim_{t\to 0}\frac{U(t)-\mathbbm{1}}{t}\psi\in\cH \right.\right\} \\
    A\psi\;&:=\;\ii\frac{\ud}{\ud t}\Big|_{t=0}U(t)\psi\,.
   \end{split}
  \end{equation*}
The operator $A$ is hermitian, for 
\[
 \langle\psi,A\psi\rangle\;=\;\lim_{t\to 0}\left\langle\psi,\ii\frac{U(t)-\mathbbm{1}}{t}\psi\right\rangle\;=\;\lim_{t\to 0}\left\langle-\ii\frac{U(-t)-\mathbbm{1}}{t}\psi,\psi\right\rangle\;=\;\langle A\psi,\psi\rangle
\]
for every $\psi\in\mathcal{D}(A)$.

For each $\psi\in\mathcal{D}(H)$ one has
\[
 \Big\|\,\ii\frac{U(t)-\mathbbm{1}}{t}\psi-H\psi\Big\|^2\;=\;\int_{\mathbb{R}}\Big|\,\ii\frac{e^{-\ii t\lambda}-1}{t}-\lambda\Big|^2\ud\mu_{\psi}^{(H)}(\lambda)\;\xrightarrow[]{\:t\to 0\:}\;0\,,
\]
where $\mu_{\psi}^{(H)}$ is the scalar spectral measure of $H$ relative to the vector $\psi$ (see, e.g., \cite[Lemma 4.4 and Theorem 5.7]{schmu_unbdd_sa}): the above identity is an immediate consequence of the properties of the functional calculus, whereas the limit as $t\to 0$ follows from dominated convergence and mean value theorem. Thus, $\psi\in\mathcal{D}(A)$ and
\[
 H\psi\;=\;\ii\frac{\ud}{\ud t}\Big|_{t=0}U(t)\psi\;=\;A\psi\,.
\]
This means that the hermitian operator $A$ is an extension of the self-adjoint operator $H$: by maximality of symmetry of any self-adjoint operator (see, e.g., \cite[Sect.~3.2]{schmu_unbdd_sa}), necessarily $A=H$.

Since $A$ is the same as $H$, \eqref{eq:actdomH} is therefore established. Moreover, for $t\in\mathbb{R}$ and $\psi_0\in\mathcal{D}(H)$ one has
\[
 \ii\frac{\ud}{\ud t}U(t)\psi_0\;=\;\ii\lim_{\tau\to 0}\frac{U(t+\tau)-U(t)}{\tau}\psi_0\;=\;\ii \,U(t)\lim_{\tau\to 0}\frac{U(\tau)-\mathbbm{1}}{\tau}\psi_0\;=\;U(t)H\psi_0\,.
\]
This, together with the already proved identity $HU(t)\psi_0=U(t)H\psi_0$, establishes \eqref{eq:HU=UH=dU}. Properties 1.~and 2.~are thus proved. In turn, \eqref{eq:HU=UH=dU} yields finally \eqref{eq:SchrIVP}.

Concerning property 3., \eqref{eq:SchrIVP} clearly implies that $\psi(\cdot)\in C^1(\mathbb{R},\cH)$ (in fact, with values in $\mathcal{D}(H)$ for every $t\in\mathbb{R}$). Should there exist two such solutions $\psi_1(\cdot)$ and $\psi_2(\cdot)$, then $\phi(t):=\psi_1(t)-\psi_2(t)$ would satisfy $\phi(0)=0$ as well as
\[
\begin{split}
  \frac{\ud}{\ud t}\|\phi(t)\|^2\;&=\;\Big\langle \frac{\ud}{\ud t}\phi(t),\phi(t)\Big\rangle+\Big\langle\phi(t),\frac{\ud}{\ud t}\phi(t)\Big\rangle \\
  &=\;\langle-\ii H\phi(t),\phi(t)\rangle+\langle\phi(t),-\ii H\phi(t)\rangle\;=\;0\,.
\end{split}
\]
Thus, $\|\phi(t)\|=\|\phi(0)\|=0$, meaning $\psi_1(t)=\psi_2(t)$ for every $t\in\mathbb{R}$.

\underline{Implication (i) $\Rightarrow$ (ii)}. Owing to Stone's theorem, $U(t)=e^{-\ii t A}$ for every $t\in\mathbb{R}$, where $A$ is a uniquely determined self-adjoint operator. By the very same arguments developed in the first part of the proof (replacing now $H$ with $A$), the domain and the action of $A$ are given by \eqref{eq:defA} (which is not a definition now), and moreover
\[
 \ii\frac{\ud}{\ud t}U(t)\psi_0\;=\;A U(t)\psi_0\qquad\forall \psi_0\in\mathcal{D}(A)\,,\;\forall t\in\mathbb{R}\,.
\]
This and \eqref{eq:SchrIVP} then imply that each $\psi_0$ from $\mathcal{D}(A)$ also belongs to $\mathcal{D}(H)$, with $A\psi_0=H\psi_0$. Therefore the hermitian operator $H$ extends the self-adjoint operator $A$, which by maximality implies $H=A$. $H$ is thus necessarily self-adjoint.
\end{proof}

\begin{remark}\label{rem:uniqueness_with_hermitian}
 The above uniqueness argument $\frac{\ud}{\ud t}\|\phi(t)\|^2=0$ technically speaking involves only the \emph{hermiticity} of $H$, provided that $\phi(t)\in\mathcal{D}(H)$ for generic $t$: this is all what is needed for the step $\langle H\phi(t),\phi(t)\rangle=\langle \phi(t),H\phi(t)\rangle$. One must be guaranteed in advance, though, that both $\psi_1(t)$ and $\psi_2(t)$ evolve inside $\mathcal{D}(H)$, as indeed is assumed in the course of the proof. In the discussion of Subsect.~\ref{sec:non-uniquedynamics}, instead, we observed that no solution to $\ii\frac{\partial}{\partial t}\psi(t,x)=-\frac{\partial^2}{\partial x^2}\psi(t,x)$ exists with the property that the support of $\psi(t,\cdot)$ remains a compact in $(0,1)$ for all $t\in\mathbb{R}$. That is, no solution remains in the domains of the hermitian-only operator $H_\circ$: the uniqueness argument, in that case, concerns a non-existing solution.
\end{remark}

\begin{proof}[Proof of Theorem \ref{thm:onbEV-sa}]~

 Let $(\psi_n)_{n\in\mathbb{N}}$ be the considered orthonormal basis of eigenvectors of $A$ and let $(\lambda_n)_{n\in\mathbb{N}}$ be the collection of the corresponding eigenvalues, all counted with multiplicity. As $A$ is hermitian, the $\lambda_n$'s are all real. Besides, since $(\psi_n)_{n\in\mathbb{N}}$ is an orthonormal basis for $\cH$ and is contained in $\mathcal{D}(A)$, then $\mathcal{D}(A)$ is dense in $\cH$.

 Concerning the subspace $\mathrm{ran}(A+\ii\mathbbm{1})$ (the range of the operator $A+\ii\mathbbm{1}$), we see that
 \[
  \psi_n\;=\;(\lambda_n+\ii)^{-1}(A+\ii\mathbbm{1})\psi_n\;\in\;\mathrm{ran}(A+\ii\mathbbm{1})\,,
 \]
 hence $\mathrm{ran}(A+\ii\mathbbm{1})$ too is \emph{dense} in $\cH$.

 On the other hand, let us now see that the subspace $\mathrm{ran}(A+\ii\mathbbm{1})$ is \emph{closed} in $\cH$. That is, let $(\eta_m)_{m\in\mathbb{N}}$ be a sequence in $\mathrm{ran}(A+\ii\mathbbm{1})$ that converges to some $\eta\in\cH$ and let us show that $\eta\in \mathrm{ran}(A+\ii\mathbbm{1})$. Write $\eta_m=(A+\ii\mathbbm{1})\xi_m$ for some $\xi_m\in\mathcal{D}(A)$ and observe that, owing to the hermiticity of $A$,
 \[
  \|\eta_m-\eta_{m'}\|^2\;=\;\|(A+\ii\mathbbm{1})(\xi_m-\xi_{m'})\|^2\;=\;\|A(\xi_m-\xi_{m'})\|^2+\|\xi_m-\xi_{m'}\|^2\,.
 \]
 Thus,
 \[
  \|\xi_m-\xi_{m'}\|\;\leqslant\;\|\eta_m-\eta_{m'}\|\,,
 \]
 which shows that $(\xi_m)_{m\in\mathbb{N}}$ is a Cauchy sequence in $\cH$ and hence, by completeness of $\cH$, converges to some $\xi\in\cH$. We thus have $\xi_m\to\xi$ and $(A+\ii\mathbbm{1})\xi_m\to\eta$ as $m\to\infty$: as $A$ is a closed operator, and so too is therefore $A+\ii\mathbbm{1}$, then necessarily $\xi\in\mathcal{D}(A)$ and $(A+\ii\mathbbm{1})\xi=\eta$. This shows precisely that $\eta\in\mathrm{ran}(A+\ii\mathbbm{1})$.

 The range of $(A+\ii\mathbbm{1})$ being simultaneously a dense and closed subspace of $\cH$, one concludes that $\mathrm{ran}(A+\ii\mathbbm{1})=\cH$. Analogously, $\mathrm{ran}(A-\ii\mathbbm{1})=\cH$.

 For a generic densely defined operator $T$ on Hilbert space, one has $\ker T^\dagger=(\mathrm{ran}T)^\perp$, i.e., the orthogonal complement to the range of $T$ is precisely the kernel of $T^\dagger$ (see, e.g., \cite[Proposition 1.6(ii)]{schmu_unbdd_sa}). In the present case, with $T\equiv (A+\ii\mathbbm{1})$ and hence $T^\dagger=(A^\dagger-\ii\mathbbm{1})$, the fact that $\mathrm{ran}(A+\ii\mathbbm{1})=\cH$ implies $\ker(A^\dagger-\ii\mathbbm{1})=\{0\}$.

 $A$ is densely defined and hermitian, therefore $\mathcal{D}(A)\subset\mathcal{D}(A^\dagger)$, with $A$ and $A^\dagger$ giving the same output on elements of $\mathcal{D}(A)$. Let us now show that $\mathcal{D}(A)\supset\mathcal{D}(A^\dagger)$: this would imply $\mathcal{D}(A^\dagger)=\mathcal{D}(A)$ and therefore $A^\dagger=A$.

 To this aim, let $\phi\in\mathcal{D}(A^\dagger)$ and let us show that $\phi\in\mathcal{D}(A)$. Since $\mathrm{ran}(A-\ii\mathbbm{1})=\cH$, then $(A^\dagger-\ii\mathbbm{1})\phi=(A-\ii\mathbbm{1})\psi$ for some $\psi\in\mathcal{D}(A)$. The latter identity, owing to the fact that $\psi\in\mathcal{D}(A)\subset \mathcal{D}(A^\dagger)$, implies
 \[
  (A^\dagger-\ii\mathbbm{1})(\phi-\psi)\;=\;0\,,
 \]
 i.e., $\phi-\psi\in\ker(A^\dagger-\ii\mathbbm{1})$. But, as seen above, $\ker(A^\dagger-\ii\mathbbm{1})=\{0\}$. Then $\phi=\psi\in\mathcal{D}(A)$. 
\end{proof}

 \begin{proof}[Proof of Theorem \ref{thm:closed-sa}]~
  
  By assumption, for some $m\in\mathbb{R}$, one has $\mathcal{E}[\psi]\geqslant m\|\psi\|^2$ $\forall\psi\in\mathcal{D}[\mathcal{E}]$. Without loss of generality one can assume that the form $\mathcal{E}$ has lower bound $m=1$, because $\psi\mapsto\mathcal{E}[\psi]$ is closed if and only if $\psi\mapsto\mathcal{E}[\psi]+\lambda\|\psi\|^2$ is so (on the same form domain), irrespectively of $\lambda\in\mathbb{R}$, and the corresponding operators defined by  \eqref{eq:fromFormToOp} are $A$ and $A+\lambda\mathbbm{1}$, where one is self-adjoint if and only if so is the other (on the same operator domain). Thus, non-restrictively, let us set $m=1$.

  %Let now $A$ be the operator associated with $\mathcal{E}$ through definition \eqref{eq:fromFormToOp}, and 
  Define now
  \[
   \|\psi\|_{\mathcal{E}}\;:=\;\mathcal{E}[\psi]^{1/2}\,,\qquad \langle\psi,\phi\rangle_{\mathcal{E}}\;:=\;\mathcal{E}[\psi,\phi]
  \]
  on vectors from $\mathcal{D}[\mathcal{E}]$. These are respectively a norm and the corresponding scalar product. Moreover, as a consequence of the closedness of $\mathcal{E}$, the space $\mathcal{D}[\mathcal{E}]$ equipped with the scalar product $\langle\cdot,\cdot\rangle_{\mathcal{E}}$ is a Hilbert space. Indeed, if $(\psi_n)_{n\in\mathbb{N}}$ is a Cauchy sequence in $(\mathcal{D}[\mathcal{E}],\langle\cdot,\cdot\rangle_{\mathcal{E}})$, then the inequality $\|\psi_n-\psi_m\|_{\mathcal{E}}\geqslant\|\psi_n-\psi_m\|$ implies that it is also a Cauchy sequence in $\cH$ and therefore $\|\psi_n-\psi\|\to 0$ for some $\psi\in\cH$. The definition of closed form given in Subsect.~\ref{sec:closedsemibdd-sa} now implies that $\psi\in\mathcal{D}[\mathcal{E}]$ and $\|\psi_n-\psi\|_{\mathcal{E}}\to 0$. This proves that Cauchy sequences in $(\mathcal{D}[\mathcal{E}],\langle\cdot,\cdot\rangle_{\mathcal{E}})$ converge and therefore that the latter is indeed a Hilbert space.

  Let $A$ be the operator associated with the form $\mathcal{E}$ according to the definition \eqref{eq:fromFormToOp}. $A$ is hermitian, because the form $\mathcal{E}$ is symmetric. Next, let us show that $\mathrm{ran}A=\cH$.

  To this aim, pick an arbitrary $\xi_0\in\cH$. The map $\psi\mapsto\langle\xi_0,\psi\rangle$ is a linear functional on $(\mathcal{D}[\mathcal{E}],\langle\cdot,\cdot\rangle_{\mathcal{E}})$ which is also continuous, since $\|\cdot\|\leqslant\|\cdot\|_{\mathcal{E}}$. As every continuous linear functional on Hilbert space, it can be represented by means of Riesz theorem (see, e.g.,  \cite[Theorem II.4]{rs1}) as
  \[
   \langle\xi_0,\psi\rangle\;=\;\langle\psi_0,\psi\rangle_{\mathcal{E}}\;=\;\mathcal{E}[\psi_0,\psi]\qquad\forall\psi\in\mathcal{D}[\mathcal{E}]
  \]
  for some $\psi_0\in \mathcal{D}[\mathcal{E}]$. Then \eqref{eq:fromFormToOp} says that $\psi_0\in\mathcal{D}(A)$ and $A\psi_0=\xi_0$. Thus, $\mathrm{ran}A=\cH$.

  The latter property also allows one to deduce that the subspace $\mathcal{D}(A)$ is dense in $\cH$. Indeed, for a generic $\phi_0\perp\mathcal{D}(A)$ one can write $\phi_0=A\psi_0$ for some $\psi_0\in\mathcal{D}(A)$, whence
  \[
   0\;=\;\langle\phi_0,\psi\rangle\;=\;\langle A\psi_0,\psi\rangle\;=\;\langle\psi_0,A\psi\rangle\qquad\forall\psi\in\mathcal{D}(A)\,.
  \]
  But this means $\psi_0\perp \mathrm{ran}A$, so the only possibility is $\phi_0=0$. $\mathcal{D}(A)$ is therefore dense.

  From now on, let us reason as done in the above proof of Theorem \ref{thm:onbEV-sa}, in order to show that $A=A^\dagger$. From the fact that $\mathrm{ran}A=\cH$ one concludes that $\ker A^\dagger=(\mathrm{ran}A)^\perp=\{0\}$, i.e., $A^\dagger$ injective. Let $\phi\in\mathcal{D}(A^\dagger)$ and let us show that $\phi\in\mathcal{D}(A)$. Since $\mathrm{ran}A=\cH$, then $A^\dagger\phi=A\psi$ for some $\psi\in\mathcal{D}(A)$. The latter identity, owing to the fact that $\psi\in\mathcal{D}(A)\subset \mathcal{D}(A^\dagger)$, implies $A^\dagger(\phi-\psi)=0$. By injectivity of $A^\dagger$, $\phi=\psi\in\mathcal{D}(A)$. 
   \end{proof}

% \bibliographystyle{siam}
% \bibliography{bib_ALE}

\def\cprime{$'$}

\end{document}